\documentclass[copyright,creativecommons]{eptcs}
\providecommand{\event}{ACT 2023} 

\usepackage{amsmath}
\usepackage{amssymb}
\usepackage{amsthm}
\usepackage{dsfont}
\usepackage{mathrsfs}
\usepackage{mathtools}
\usepackage{stmaryrd}
\usepackage[all]{xy}
\usepackage[utf8]{inputenc} 
\usepackage{hyperref}
\usepackage{tikz}
\usetikzlibrary{arrows.meta}

\title{Subsumptions of Algebraic Rewrite Rules}

\author{Thierry~Boy~de~la~Tour
\institute{Univ. Grenoble Alpes, CNRS, Grenoble INP, LIG, 38000 Grenoble, France}
\email{thierry dot boy-de-la-tour at imag.fr}
}
\def\titlerunning{Subsumptions of Algebraic Rewrite Rules}
\def\authorrunning{T. Boy de la Tour}

\theoremstyle{plain}
\newtheorem{theorem}{Theorem}[section]
\newtheorem{lemma}[theorem]{Lemma}
\newtheorem{proposition}[theorem]{Proposition}
\newtheorem{corollary}[theorem]{Corollary}

\theoremstyle{definition}
\newtheorem{definition}[theorem]{Definition}
\newtheorem{example}[theorem]{Example}

\newcommand{\tuple}[1]{(#1)}

\newcommand{\invf}[1]{#1^{-1}}

\newcommand{\Gc}{\mathrm{G}}
\newcommand{\Lc}{\mathrm{L}}
\newcommand{\Dc}{\mathrm{D}}
\newcommand{\Kc}{\mathrm{K}}
\newcommand{\Rc}{\mathrm{R}}
\newcommand{\fc}{\mathrm{f}}
\newcommand{\kc}{\mathrm{k}}
\newcommand{\rc}{\mathrm{r}}
\newcommand{\lc}{\mathrm{l}}

\newcommand{\mc}{\mathrm{m}}

\newcommand{\projd}{\pi_1}
\newcommand{\projk}{\pi_2}
\newcommand{\projr}{\pi_3}
\newcommand{\rulecatDPO}{\mathscr{R}_{\mathrm{DPO}}}
\newcommand{\rulecatDPOm}{\mathscr{R}_{\mathrm{DPOm}}}
\newcommand{\rulecatSqPO}{\mathscr{R}_{\mathrm{SqPO}}}
\newcommand{\rulecatSqPOm}{\mathscr{R}_{\mathrm{SqPOm}}}
\newcommand{\rulecatPBPO}{\mathscr{R}_{\mathrm{PBPO}}}
\newcommand{\directcatDPO}{\mathscr{D}_{\mathrm{DPO}}}
\newcommand{\directcatSqPO}{\mathscr{D}_{\mathrm{SqPO}}}
\newcommand{\directcatSqPOm}{\mathscr{D}_{\mathrm{SqPOm}}}
\newcommand{\directcatDPOm}{\mathscr{D}_{\mathrm{DPOm}}}
\newcommand{\directcatPBPO}{\mathscr{D}_{\mathrm{PBPO}}}
\newcommand{\emptygr}{\varnothing}

\newcommand{\DPORF}{\mathsf{R}_{\mathrm{DPO}}}
\newcommand{\DPOPF}{\mathsf{P}_{\mathrm{DPO}}}
\newcommand{\DPOmRF}{\mathsf{R}_{\mathrm{DPOm}}}
\newcommand{\SqPORF}{\mathsf{R}_{\mathrm{SqPO}}}
\newcommand{\SqPOmRF}{\mathsf{R}_{\mathrm{SqPOm}}}
\newcommand{\DPOmPF}{\mathsf{P}_{\mathrm{DPOm}}}
\newcommand{\PBPORF}{\mathsf{R}_{\mathrm{PBPO}}}

\newcommand{\id}[1]{1_{#1}}
\newcommand{\dual}[1]{#1^{\mathrm{op}}}
\newcommand{\arule}{\rho}
\newcommand{\rulem}{\sigma}
\newcommand{\aFunc}{\mathsf{F}}
\newcommand{\discr}[1]{|#1|}
\newcommand{\tm}[1]{#1!}

\newcommand{\cat}{\mathcal{C}}
\newcommand{\directcat}{\mathcal{D}}
\newcommand{\directobj}{\delta}
\newcommand{\directmorph}{\mu}
\newcommand{\rulecat}{\mathcal{R}}
\newcommand{\partialcat}{\cat_{\mathrm{pt}}}
\newcommand{\ptobj}{\tau}
\newcommand{\ptmorph}{\nu}
\newcommand{\inputF}{\mathsf{In}}
\newcommand{\ruleF}{\mathsf{R}}
\newcommand{\partialF}{\mathsf{P}}
\newcommand{\insertF}{\mathsf{I}}
\newcommand{\insertGF}{\mathsf{I}_{G}}
\newcommand{\insertSF}{\mathsf{I}_{\rulesys}}
\newcommand{\rulesys}{\mathcal{S}}
\newcommand{\termcat}{\mathbf{1}}
\newcommand{\directcatG}{\directcat|_{G}}
\newcommand{\directcatSG}{\directcat|_{G}^{\rulesys}}

\newcommand{\ruleFSG}{\ruleF|_{G}^{\rulesys}}
\newcommand{\stratcat}{\mathrm{\Delta}}
\newcommand{\ProjG}{\mathsf{P}^{\leftarrow}_{\stratcat}}
\newcommand{\ProjH}{\mathsf{P}^{\rightarrow}_{\stratcat}}
\newcommand{\slice}[1]{\cat\setminus #1}
\newcommand{\coslice}[1]{#1\setminus \cat}
\newcommand{\interf}{\mathrm{C}_{\stratcat}}

\newcommand{\cone}[1]{\gamma_{#1}}
\newcommand{\intercone}{\gamma_{\stratcat}}

\newcommand{\pbporulecat}{\mathscr{D}_{\mathrm{PBPO}}}

\newcommand{\egr}{\raisebox{-0.9ex}{
  \begin{tikzpicture}[xscale=0.8]
    \coordinate (A) at (0,0); \coordinate (B) at (1,0);
    \draw (A) node {$\bullet$}; \draw (B) node {$\bullet$};
    \path[-Stealth] (A) edge (B);
  \end{tikzpicture}}}
\newcommand{\dvgr}{\raisebox{-0.9ex}{
  \begin{tikzpicture}[xscale=0.8]
    \coordinate (A) at (0,0); \coordinate (B) at (1,0);
    \draw (A) node {$\bullet$}; \draw (B) node {$\bullet$};
  \end{tikzpicture}}}
\newcommand{\tvgr}{\raisebox{-0.9ex}{
  \begin{tikzpicture}[xscale=0.8]
    \coordinate (A) at (0,0); \coordinate (B) at (2,0); \coordinate (C) at (1,0);
    \draw (A) node {$\bullet$}; \draw (B) node {$\bullet$}; \draw (C) node {$\bullet$};
  \end{tikzpicture}}}
\newcommand{\degr}{\raisebox{-0.9ex}{
  \begin{tikzpicture}[xscale=0.8]
    \coordinate (A) at (0,0); \coordinate (B) at (1,0);
    \draw (A) node {$\bullet$}; \draw (B) node {$\bullet$};
    \path[-Stealth,bend right] (A) edge (B);
    \path[-Stealth,bend left] (A) edge (B);
  \end{tikzpicture}}}
\newcommand{\tvdegr}{\raisebox{-0.9ex}{
  \begin{tikzpicture}[xscale=0.8]
    \coordinate (A) at (0,0); \coordinate (B) at (2,0); \coordinate (C) at (1,0);
    \draw (A) node {$\bullet$}; \draw (B) node {$\bullet$}; \draw (C) node {$\bullet$};
    \path[-Stealth] (A) edge (C);
    \path[-Stealth] (C) edge (B);
  \end{tikzpicture}}}
\newcommand{\tvqegr}{\raisebox{-0.9ex}{
  \begin{tikzpicture}[xscale=0.8]
    \coordinate (A) at (0,0); \coordinate (B) at (2,0); \coordinate (C) at (1,0);
    \draw (A) node {$\bullet$}; \draw (B) node {$\bullet$}; \draw (C) node {$\bullet$};
    \path[-Stealth,bend left] (A) edge (C);
    \path[-Stealth,bend right] (A) edge (C);
    \path[-Stealth,bend left] (C) edge (B);
    \path[-Stealth,bend right] (C) edge (B);
  \end{tikzpicture}}}

\begin{document}
\maketitle
\begin{abstract}
  What does it mean for an algebraic rewrite rule to subsume another
  rule (that may then be called a subrule)? We view subsumptions as
  rule morphisms such that the simultaneous application of a rule and
  a subrule (i.e. the application of a subsumption morphism) yields the same
  result as a single application of the subsuming rule. Simultaneous
  applications of categories of rules are obtained by Global
  Coherent Transformations and illustrated on graphs in the DPO
  approach. Other approaches are possible since these transformations
  are formulated in an abstract Rewriting Environment, and such
  environments exist for various approaches to Algebraic Rewriting,
  including DPO, SqPO and PBPO.
\end{abstract}

\section{Introduction}

In Global Transformations \cite{MaignanS15} rules may be seen as pairs
$\tuple{L,R}$ of graphs (or objects in a category $\cat$) that are applied
simultaneously to an input graph (as in L-systems \cite{FernandezMS22}
and cellular automata \cite{FernandezMS21}). Such rules are related by
pairs of $\cat$-morphisms. These morphisms come from representing
possible overlaps of rules as subrules whose applications are induced
by the overlapping applications of rules, therefore establishing a
link between these. By computing a colimit of a diagram involving the
morphisms between occurrences of right-hand sides, Global
Transformations offer the possibility to merge items (vertices or
edges) in these occurrences of right-hand sides.

This form of rules has the advantage of simplicity, first because rule
morphisms are those of the product category $\cat\times\cat$, and
second because the input object is completely removed. Indeed, when all
occurrences of $L$ have been found in the input graph $G$, the output
graph $H$ is produced solely from the corresponding occurrences of
$R$, thus effectively removing $G$. In particular, if no $L$ has any
match in $G$ then $H$ is the empty graph. If $G$ is, say, a relational
database, this may be inconvenient.

More standard approaches to algebraic rewriting use rules for
\emph{replacing} matched parts of the input object by new parts. These
substitutions are performed by first removing the matched part and
then adding the new part, this last operation being performed by a
pushout. But since there is no general algebraic way of removing parts
of a $\cat$-object, several approaches have been devised, from DPO
\cite{Ehrig79} to PBPO \cite{CorradiniDEPR19} rules, for defining the
\emph{context} (a $\cat$-object) in which $R$ can be ``pushed''. These rules
always have an interface $K$ with a pair of $\cat$-morphisms from $K$
to $L$ and $R$ (a span), but can be more complicated. Hence the
necessity of a general notion of morphism between rules that does not
depend on a specific shape of rules.

In Section \ref{sec-DPO} an intuitive analysis of rule subsumptions on
a simple example with DPO-rules leads to a natural definition of
subsumption morphisms between DPO-rules, and of corresponding
subsumption morphisms between direct DPO-transformations. This leads
in Section~\ref{sec-RE} to a general notion of \emph{Rewriting
  Environment} that provides the relevant categories of rules and of
direct transformations, and functors between them and to a category of
\emph{partial transformations}.

Section~\ref{sec-GCT} is devoted to the Global Coherent
Transformation.  It derives from the Parallel Coherent Transformations
defined in \cite{BdlTE21a} (only for a variant of DPO-rules), where
sets or rules can be applied simultaneously on an input object. The
first step defines the \emph{global context} as a limit of a diagram
that involves the subsumption morphisms.

One important problem is that overlapping applications of rules (i.e.,
overlapping direct transformations) may conflict as one transformation
deletes an item of $G$ that another transformation
preserves. Note that conflicts cannot happen with Global
Transformations since they preserve nothing. Only non conflicting, so
called \emph{coherent} transformations can be applied simultaneously,
hence the notion of Parallel Coherence from \cite{BdlTE21a} must be
adapted in order to embrace subsumption morphisms. The adapted
definition ensures that the right-hand sides of the rules can be
pushed in the global context by means of a colimit.

Section \ref{sec-env} is devoted to the analysis of Rewriting
Environments, and yields natural definitions of environments for the
SqPO and PBPO approaches. Future work and open questions are found in
Section~\ref{sec-conclusion}.

\section{Notations}\label{sec-notations}

\emph{Embeddings} are injective functors, all
other notions are compatible with \cite{MacLane}. We also use
\emph{meets} and \emph{sums} of functors, see \cite{Lawvere93}.

For any category $\cat$, we write $G\in\cat$ to indicate that $G$ is a
$\cat$-object, and $\discr{\cat}$ is the discrete category on
$\cat$-objects. Then $G$ also denotes the functor from the terminal
category $\termcat$ to $\discr{\cat}$ that maps the object of
$\termcat$ to $G$. $\emptygr$ denotes the initial object of $\cat$, if
any. The \emph{slice} category $\slice{G}$ has as objects
$\cat$-morphisms of codomain $G$, and as morphisms $h:f\rightarrow g$
$\cat$-morphisms such that $g\circ h = f$. The \emph{coslice} category
$\coslice{G}$ has as objects $\cat$-morphisms of domain $G$, and as
morphisms $h:f\rightarrow g$ $\cat$-morphisms such that
$h\circ f = g$.

We will use the standard notion of graphs with multiple directed
edges.  
In the running example we will use graphs with 2 to 3 vertices and 0
to 4 edges denoted directly by their drawings, as in $\dvgr$ and
$\tvqegr$. In order to avoid naming vertices, they will always be
depicted from left to right, and we will use at most two monomorphisms
from one graph to another: one (depicted as a plain arrow) that maps
the leftmost (resp.  rightmost) vertex of the domain graph to the
leftmost (resp.  rightmost) vertex of the codomain graph, and one
(dotted arrow) that swaps these vertices. For example we consider only
two possible morphisms:
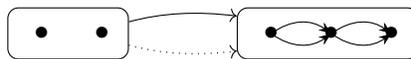
\begin{center}
    \begin{tikzpicture}[xscale=3]
    \node[draw,rounded corners](D) at (0,0) {$\dvgr$};
    \node[draw,rounded corners](C) at (1.15,0){$\tvqegr$};
    \path[->,bend left] (D) edge (C);
    \path[->,bend right, dotted] (D) edge (C);
  \end{tikzpicture}
\end{center}

The two morphisms from $\egr$ to $\degr$ will be distinguished
similarly:
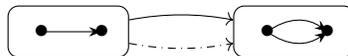
\begin{center}
  \begin{tikzpicture}[xscale=3]
    \node[draw,rounded corners](D) at (0,0) {$\egr$};
    \node[draw,rounded corners](C) at (1,0){$\degr$};
    \path[->,bend left] (D) edge (C);
    \path[->,bend right,dash dot] (D) edge (C);
  \end{tikzpicture}
\end{center}

\section{Subrules in DPO Graph Transformations}\label{sec-DPO}

The notion of a rule $\arule$ being a subrule of a rule $\arule'$, or
more generally of a subsumption morphism $\rulem:\arule\rightarrow \arule'$,
covers the idea that $\arule$ represents a part (specified by
$\rulem$) of what $\arule'$ achieves, and therefore that any
application of $\arule'$ entails and subsumes a particular application
(obtained through $\rulem$) of $\arule$. We first try to make this
idea more precise with DPO-rules.

\begin{definition}[DPO rules and direct transformations, gluing condition]
\label{def-ruleDPO}
  A \emph{DPO-rule} $\arule$ in a category $\cat$ is a span diagram
  \begin{center}
    \begin{tikzpicture}[xscale=1.5]
      \node (L) at (1,1) {$L$};
      \node (K) at (2,1) {$K$};  \node (R) at (3,1) {$R$};  
      \path[>->] (K) edge node[above, font=\footnotesize] {$l$} (L);
      \path[->] (K) edge node[above, font=\footnotesize] {$r$} (R);
    \end{tikzpicture}
  \end{center}
  in $\cat$, where $l$ is monic. Diagrams in $\cat$ are functors from
  an index category to $\cat$, and it will sometime be convenient to
  refer to the objects and morphisms of this index category; they will
  be denoted by the corresponding roman letters (here $\arule\Lc = L$,
  $\arule\lc = l$, etc.)

  We say that an item (edge or vertex) of a graph $G$ is \emph{marked
    for removal} by a \emph{matching} $m:L \rightarrow G$ for a rule
  $\arule$ if it has a preimage by $m$ that has none by $l$ (see
  \cite{BdlT23a}). The \emph{gluing condition} for $m,\,\arule$ states
  that
  \[\left\{
    \begin{array}{lc}
      \text{all items marked for removal have only one preimage by } m,
      &\text{(GC1)}\\
      \text{if a vertex adjacent to an edge is marked for removal, then so
      is this edge.} & \text{(GC2)}
    \end{array}\right.\]

  A \emph{direct DPO-transformation} $\directobj$ in $\cat$ is a diagram 
  \begin{center}
    \begin{tikzpicture}[xscale=1.7,yscale=1.5]
      \node (D1) at (0,1) {$L$};
      \node (K1) at (1,1) {$K$};  \node (R1) at (2,1) {$R$};  
      \node (D2) at (0,0) {$G$};
      \node (K2) at (1,0) {$D$};  \node (R2) at (2,0) {$H$};  
      \path[>->] (K1) edge node[above, font=\footnotesize] {$l$} (D1);
      \path[->] (K1) edge node[above, font=\footnotesize] {$r$} (R1);
      \path[->] (K2) edge node[below, font=\footnotesize] {$f$} (D2);
      \path[->] (K2) edge node[below, font=\footnotesize] {$g$} (R2);
      \path[->] (D1) edge node[left, font=\footnotesize] {$m$} (D2);
      \path[->] (K1) edge node[right, font=\footnotesize] {$k$} (K2);
      \path[->] (R1) edge node[right, font=\footnotesize] {$n$} (R2);
      \draw (0.3,0.1) to (0.3,0.3) to (0.1,0.3);
      \draw (1.7,0.1) to (1.7,0.3) to (1.9,0.3);
    \end{tikzpicture}
  \end{center}
  in $\cat$ such that $l$ is monic and the two squares are pushouts.
\end{definition}

It is well known (see \cite{EhrigEPT06,HndBkCorradiniMREHL97}) that in
the category of graphs, given $\arule$ and $m:L\rightarrow G$, there
exists a direct DPO-transformation $\directobj$ with $\arule$ and $m$
iff the gluing condition holds. The \emph{pushout complement} $D$ is
then a subgraph of $G$ ($f$ is monic) and contains all the items of
$G$ that are not marked for removal.

\begin{example}\label{ex-running}
  In the running example we transform every directed edge in a graph
into a pair of consecutive edges. This can be expressed as the
following rule 
\begin{equation*}\tag{$\arule'$}
  \raisebox{-1ex}{\begin{tikzpicture}[xscale=2.5]
    \node[draw,rounded corners](L) at (0,0) {$\egr$};
    \node[draw,rounded corners](K) at (1,0){$\dvgr$};
    \node[draw,rounded corners](R) at (2.15,0){$\tvdegr$};
    \path[->](K) edge (L); \path[->](K) edge (R);
  \end{tikzpicture}}
\end{equation*}

We do not wish to transform loops in this way, hence we adopt the DPO
approach restricted to monic matchings. We also wish to create only one
middle vertex for parallel edges, so that the input graph $G=\degr$ in
our running example shall be transformed into $H=\tvqegr$. In order to
merge the two vertices created by the two simultaneous applications of $\arule'$ on
$G$ we need to link them through the application of a common subrule
on their overlap. Consider the rule
\begin{equation*}\tag{$\arule$}
  \raisebox{-1ex}{\begin{tikzpicture}[xscale=2.5]
    \node[draw,rounded corners](L) at (0,0) {$\dvgr$};
    \node[draw,rounded corners](K) at (1,0){$\dvgr$};
    \node[draw,rounded corners](R) at (2.15,0){$\tvgr$};
    \path[->](K) edge (L); \path[->](K) edge (R);
  \end{tikzpicture}}
\end{equation*}
The right hand side expresses the fact that the middle vertex is
created depending on the overlap $\dvgr$ and not on the edges of
$G$. Thus we need to link the middle vertices from $\arule$ and
$\arule'$ right-hand sides through a morphism
$\rulem^+:\arule\rightarrow \arule'$, given as three $\cat$-morphisms:
\begin{equation*}\tag{$\rulem^+$}
  \raisebox{-5ex}{\begin{tikzpicture}[xscale=2.5,yscale=1.5]
    \node[draw,rounded corners](L) at (0,0) {$\dvgr$};
    \node[draw,rounded corners](K) at (1,0){$\dvgr$};
    \node[draw,rounded corners](R) at (2.15,0){$\tvgr$};
    \path[->](K) edge (L); \path[->](K) edge (R);
    \node[draw,rounded corners](L') at (0,1) {$\egr$};
    \node[draw,rounded corners](K') at (1,1){$\dvgr$};
    \node[draw,rounded corners](R') at (2.15,1){$\tvdegr$};
    \path[->](K') edge (L'); \path[->](K') edge (R');
    \path[->](L) edge node[left, font=\footnotesize] {$\rulem^+_1$} (L'); 
    \path[->](K) edge node[right, font=\footnotesize] {$\rulem^+_2$} (K'); 
    \path[->](R) edge node[right, font=\footnotesize] {$\rulem^+_3$} (R'); 
  \end{tikzpicture}}
\end{equation*}
\end{example}

The two square diagrams commute, and we easily understand that this is
necessary for $\arule$ to be a subrule of $\arule'$. But commutation
would also hold if the interface graph of $\arule$ were $\emptygr$,
and then $\arule$ would remove the overlap $\dvgr$. This would
conflict with $\arule'$ that preserves this part of $G$. We need the two
rules to behave similarly on the overlap, which means that the
interface of the subrule $\arule$ is determined by the way the
interface of $\arule'$ intersects the overlap. This can be expressed
by stating that the left square should be a pullback.

\begin{definition}[categories $\rulecatDPO$,
  $\rulecatDPOm$]\label{def-catDPO}
  For any category $\cat$, let $\rulecatDPO$ be the category whose
  objects are the DPO-rules and morphisms (or \emph{subsumptions})
  $\rulem:\arule\rightarrow\arule'$ are triples
  $\rulem=\tuple{\rulem_1,\rulem_2,\rulem_3}$ of $\cat$-morphisms such
  that
  \begin{center}
    \begin{tikzpicture}[xscale=1.7,yscale=1.5]
      \node (D1) at (1,1) {$L$};
      \node (K1) at (2,1) {$K$};  \node (R1) at (3,1) {$R$};  
      \node (D2) at (1,0) {$L'$};
      \node (K2) at (2,0) {$K'$};  \node (R2) at (3,0) {$R'$};  
      \path[>->] (K1) edge node[above, font=\footnotesize] {$l$} (D1);
      \path[->] (K1) edge node[above, font=\footnotesize] {$r$} (R1);
      \path[>->] (K2) edge node[below, font=\footnotesize] {$l'$} (D2);
      \path[->] (K2) edge node[below, font=\footnotesize] {$r'$} (R2);
      \path[->] (D1) edge node[left, font=\footnotesize] {$\rulem_1$} (D2);
      \path[->] (K1) edge node[right, font=\footnotesize] {$\rulem_2$} (K2);
      \path[->] (R1) edge node[right, font=\footnotesize] {$\rulem_3$} (R2);
      \draw (1.7,0.9) to (1.7,0.7) to (1.9,0.7);
    \end{tikzpicture}
  \end{center}
  (where $L'=\arule'\Lc$ etc.) commutes in $\cat$ and the left square
  is a pullback. Composition is componentwise
  and the obvious identities are
  $\id{\arule}=\tuple{\id{L},\, \id{K},\, \id{R}}$ (this is a
  subcategory of $\cat^{\cdot\leftarrow\cdot\rightarrow\cdot}$).  Let $\rulecatDPOm$
  be the subcategory of $\rulecatDPO$ with all rules and all morphisms
  $\rulem$ such that $\rulem_1$ and $\rulem_2$ are monics.
\end{definition}

\begin{example}
  We consider two morphisms of rules, $\rulem^+$ above and
  $\rulem^-:\arule\rightarrow\arule'$ that swaps the left and right
  vertices:
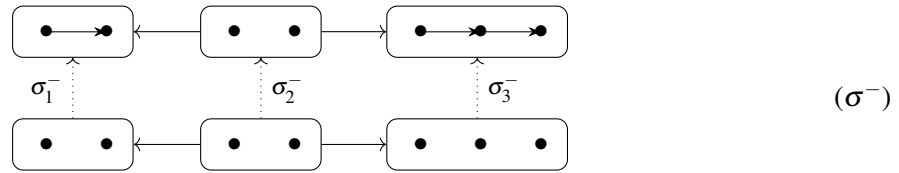
\begin{equation*}\tag{$\rulem^-$}
  \raisebox{-5ex}{\begin{tikzpicture}[xscale=2.5,yscale=1.5]
    \node[draw,rounded corners](L) at (0,0) {$\dvgr$};
    \node[draw,rounded corners](K) at (1,0){$\dvgr$};
    \node[draw,rounded corners](R) at (2.15,0){$\tvgr$};
    \path[->](K) edge (L); \path[->](K) edge (R);
    \node[draw,rounded corners](L') at (0,1) {$\egr$};
    \node[draw,rounded corners](K') at (1,1){$\dvgr$};
    \node[draw,rounded corners](R') at (2.15,1){$\tvdegr$};
    \path[->](K') edge (L'); \path[->](K') edge (R');
    \path[->,dotted](L) edge node[left, font=\footnotesize] {$\rulem^-_1$} (L'); 
    \path[->,dotted](K) edge node[right, font=\footnotesize] {$\rulem^-_2$} (K'); 
    \path[->,dotted](R) edge node[right, font=\footnotesize] {$\rulem^-_3$} (R'); 
  \end{tikzpicture}}
\end{equation*}
\end{example}

We now see that the gluing condition is inherited (backward) along
the morphisms of $\rulecatDPOm$.

\begin{proposition}\label{prop-gc}
  If $\cat$ is the category of graphs,
  $\rulem:\arule\rightarrow\arule'$ is a morphism in $\rulecatDPO$
  such that $\rulem_1$ is monic and $m':L'\rightarrow G$ satisfies the
  gluing condition for $\arule'$ then so does $m'\circ \rulem_1:
  L\rightarrow G$ for $\arule$.
\end{proposition}

\begin{example}\label{ex-pushouts}
  There are two obvious matchings $m'_1$ and $m'_2$ of $\arule'$ in
  $G$, and they induce two matchings of $\arule$ in $G$, say
  $m^+= m'_1 \circ \rulem^+_1= m'_2\circ \rulem^+_1$ and
  $m^-= m'_1 \circ \rulem^-_1= m'_2\circ \rulem^-_1$.  We see that
  $m'_1$ and $m'_2$ satisfy the gluing condition, hence they have a
  pushout complement by $l'$ and so do $m^+$ and $m^-$ by $l$. We
  therefore get two DPO-transformations of $G$ by $\arule$ (below
  left), one with $\tuple{m^+,k^+, n^+,f,g}$ and the other with
  $\tuple{m^-,k^-, n^-,f,g}$, and two DPO-transformations of $G$ by
  $\arule'$ (below right), one with $\tuple{m'_1,k', n',f'_1,g'}$ and
  the other with $\tuple{m'_2,k', n',f'_2,g'}$.

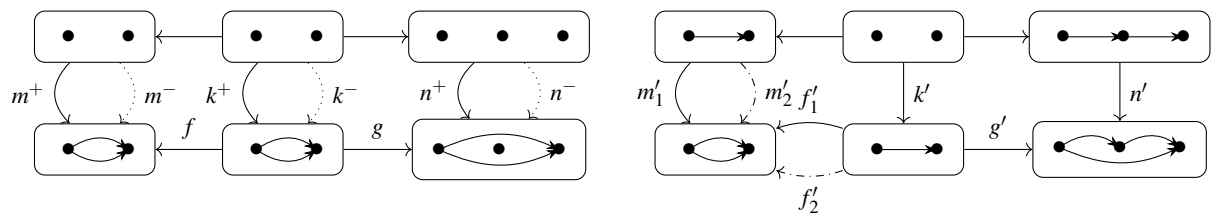
\begin{center}
  \hspace*{-2pt}\begin{tikzpicture}[xscale=2.5,yscale=1.5]
    \node[draw,rounded corners] (G1)at(3.3,0) {\degr}; 
    \node[draw,rounded corners] (L1)at(3.3,1) {\egr};
    \node[draw,rounded corners] (K1)at(4.3,1) {\dvgr}; 
    \node[draw,rounded corners] (D1)at(4.3,0) {\egr};
    \node[draw,rounded corners] (R1)at(5.45,1) {\tvdegr};
    \node[draw,rounded corners] (H1)at(5.45,0){\begin{tikzpicture}[xscale=0.8]
    \coordinate (A) at (0,0); \coordinate (B) at (2,0); \coordinate (C) at (1,0);
    \draw (A) node {$\bullet$}; \draw (B) node {$\bullet$}; \draw (C) node {$\bullet$};
    \path[-Stealth,bend left] (A) edge (C);
    \path[-Stealth,bend right=20] (A) edge (B);
    \path[-Stealth,bend left] (C) edge (B);
  \end{tikzpicture}};
    \path[->](K1) edge node[above, font=\footnotesize] {$$}(L1); 
    \path[->](K1) edge node[above, font=\footnotesize] {$$}(R1); 
    \path[->](K1) edge node[right, font=\footnotesize] {$k'$}(D1); 
    \path[->](R1) edge node[right, font=\footnotesize] {$n'$}(H1); 
    \path[->,bend right](L1) edge node[left, font=\footnotesize] {$m'_1$}(G1);
    \path[->,bend right](D1) edge node[above, font=\footnotesize] {$f'_1$}(G1);
    \path[->,bend left, dash dot](L1) edge node[right, font=\footnotesize] {$m'_2$}(G1);
    \path[->,bend left, dash dot](D1) edge node[below, font=\footnotesize] {$f'_2$}(G1);
    \path[->](D1) edge node[above, font=\footnotesize] {$g'$}(H1); 
    \node[draw,rounded corners] (G2)at(0,0) {\degr}; 
    \node[draw,rounded corners] (L2)at(0,1) {\dvgr};
    \node[draw,rounded corners] (K2)at(1,1) {\dvgr}; 
    \node[draw,rounded corners] (D2)at(1,0) {\degr};
    \node[draw,rounded corners] (R2)at(2.15,1) {\tvgr}; 
    \node[draw,rounded corners] (H2)at(2.15,0){\begin{tikzpicture}[xscale=0.8]
    \coordinate (A) at (0,0); \coordinate (B) at (2,0); \coordinate (C) at (1,0);
    \draw (A) node {$\bullet$}; \draw (B) node {$\bullet$}; \draw (C) node {$\bullet$};
    \path[-Stealth,bend left=20] (A) edge (B);
    \path[-Stealth,bend right=20] (A) edge (B);
  \end{tikzpicture}};
    \path[->](K2) edge node[above, font=\footnotesize] {$$}(L2);
    \path[->](K2) edge node[above, font=\footnotesize] {$$}(R2);
    \path[->](D2) edge node[above, font=\footnotesize] {$f$} (G2);
    \path[->](D2) edge node[above, font=\footnotesize] {$g$} (H2);
    \path[->,bend right](L2) edge node[left, font=\footnotesize] {$m^+$}(G2);
    \path[->,bend right](K2) edge node[left, font=\footnotesize] {$k^+$}(D2); 
    \path[->,bend right](R2) edge node[left, font=\footnotesize] {$n^+$}(H2); 
    \path[->,bend left,dotted](L2) edge node[right, font=\footnotesize] {$m^-$}(G2);
    \path[->,bend left,dotted](K2) edge node[right, font=\footnotesize] {$k^-$}(D2); 
    \path[->,bend left,dotted](R2) edge node[right, font=\footnotesize] {$n^-$}(H2); 
  \end{tikzpicture}
\end{center}
\end{example}

The following result reveals the relationship induced by morphisms
$\rulem:\arule\rightarrow\arule'$ on 
the corresponding direct DPO-transformations.

\begin{proposition}\label{prop-DPO}
  If $\cat$ is the category of graphs,
  $\rulem:\arule\rightarrow\arule'$ is a morphism in $\rulecatDPO$,
  $m':L'\rightarrow G$ and
  $m'\circ \rulem_1:L\rightarrow G$ have pushout complements as
  below, then there is a unique graph morphism $d$ such that
  
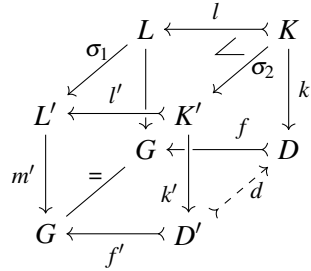
\begin{center}
    \begin{tikzpicture}[xscale=1.9,yscale=1.6,z=-0.7cm]
      \node(TK) at (0,0,0){$G$}; \node(TK') at (0,0,1){$G$};
      \node(K) at (0,1,0){$L$};
      \node(K') at (0,1,1){$L'$};
      \node(R) at (1,1,0){$K$};
      \node(R') at (1,1,1){$K'$};
      \node(TR) at (1,0,0){$D$};
      \node(TR') at (1,0,1){$D'$};
      \path[-](TK) edge node[above, font=\footnotesize] {=}(TK');
      \path[->](K) edge node[above, font=\footnotesize] {$\rulem_1$}(K');
      \path[->](R) edge node[right, font=\footnotesize] {$\rulem_2$}(R');
      \path[<-<,dashed](TR) edge node[right, font=\footnotesize] {$d$}(TR');
      \path[<-<](K) edge node[above, font=\footnotesize] {$l$}(R);
      \path[->](K) edge node[right, near start, font=\footnotesize] {$$}(TK);
      \path[->](K') edge node[left, font=\footnotesize] {$m'$}(TK');
      \path[<-<](TK) edge node[above,near end, font=\footnotesize] {$f$}(TR);
      \path[<-<](TK') edge node[below, font=\footnotesize] {$f'$}(TR');
      \path[->](R) edge node[right, font=\footnotesize] {$k$}(TR);
      \path[-](K') edge[draw=white, line width=3pt] (R');
      \path[<-<](K') edge node[above, font=\footnotesize] {$l'$}(R');
      \path[-](R') edge[draw=white, line width=3pt] (TR');
      \path[->](R') edge node[left,near end, font=\footnotesize] {$k'$}(TR');
    \draw (0.7,1,0.1) to (0.7,1,0.3) to (0.9,1,0.3);
    \end{tikzpicture}
  \end{center}
  commutes.
\end{proposition}

The existence of
$d$ means that all items marked for removal by
$m'\circ\rulem_1$, i.e., removed by the subrule
$\arule$, are also removed by
$\arule'$. In Example~\ref{ex-pushouts} we have
$f=\id{G}$, hence with $m'=m'_i$ we get
$d=f'_i$. We also see that there are no morphisms between the results
of the transformations of $G$ by $\arule$ and
$\arule'$, in either direction. This is due to the fact that subrules
remove less, but also add less. 
Subsumptions of rules cannot be deduced from the
properties of the transformation functions (from $\discr{\cat}$ to
$\discr{\cat}$) they induce.

\begin{definition}[categories $\directcatDPO$, $\directcatDPOm$,
  functors $\DPORF$, $\DPOmRF$]\label{def-transfoDPO}
  Let $\directcatDPO$ be the category whose objects are direct
  DPO-transformations in a category $\cat$ and whose morphisms (or
  \emph{subsumptions})
  $\directmorph:\directobj\rightarrow \directobj'$ are 4-tuples
  $\tuple{\directmorph_1, \directmorph_2, \directmorph_3,
    \directmorph_4}$ of $\cat$-morphisms such that the following
  diagram
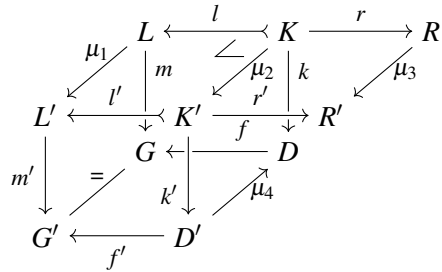
\begin{center}
  \begin{tikzpicture}[xscale=1.9,yscale=1.6,z=-0.7cm]
    \node(G) at (0,0,0){$G$}; \node(G') at (0,0,1){$G'$};
    \node(L) at (0,1,0){$L$};
    \node(L') at (0,1,1){$L'$};
    \node(K) at (1,1,0){$K$};
    \node(K') at (1,1,1){$K'$};
    \node(D) at (1,0,0){$D$};
    \node(D') at (1,0,1){$D'$};
    \node(R) at (2,1,0){$R$};
    \node(R') at (2,1,1){$R'$};
    \path[-](G) edge node[above, font=\footnotesize] {=}(G');
    \path[->](L) edge node[above, font=\footnotesize] {$\directmorph_1$}(L');
    \path[->](K) edge node[right, font=\footnotesize] {$\directmorph_2$}(K');
    \path[<-](D) edge node[right, font=\footnotesize] {$\directmorph_4$}(D');
    \path[<-<](L) edge node[above, font=\footnotesize] {$l$}(K);
    \path[->](L) edge node[right, near start, font=\footnotesize] {$m$}(G);
    \path[->](L') edge node[left, font=\footnotesize] {$m'$}(G');
    \path[<-](G) edge node[above,near end, font=\footnotesize] {$f$}(D);
    \path[<-](G') edge node[below, font=\footnotesize] {$f'$}(D');
    \path[->](K) edge node[near start,right, font=\footnotesize] {$k$}(D);
    \path[-](L') edge[draw=white, line width=3pt] (K');
    \path[<-<](L') edge node[above, font=\footnotesize] {$l'$}(K');
    \path[-](K') edge[draw=white, line width=3pt] (D');
    \path[->](K') edge node[left,near end, font=\footnotesize] {$k'$}(D');
    \path[->](K) edge node[above, font=\footnotesize] {$r$}(R);
    \path[->](R) edge node[right, font=\footnotesize] {$\directmorph_3$}(R');
    \path[-](K') edge[draw=white, line width=3pt] (R');
    \path[->](K') edge node[above, font=\footnotesize] {$r'$}(R');
    \draw (0.7,1,0.1) to (0.7,1,0.3) to (0.9,1,0.3);
  \end{tikzpicture}
\end{center}
commutes and the top left square is a pullback, with componentwise
composition (but due to the contravariance of $\directmorph_4$ we are
not in a functor category anymore). Let $\DPORF$ be the obvious
functor from $\directcatDPO$ to $\rulecatDPO$, i.e. such that
$(\DPORF \directobj)\Lc = \directobj\Lc$ etc. and
$\DPORF\directmorph = \tuple{\directmorph_1, \directmorph_2,
  \directmorph_3}$. Let $\directcatDPOm$ be the full subcategory of
$\directcatDPO$ whose objects are the direct transformations
$\directobj$ such that $\directobj\mc$ is monic, and let
$\DPOmRF: \directcatDPOm \rightarrow \rulecatDPOm$ be the
corresponding restriction of $\DPORF$.
\end{definition}

\section{Rewriting Environments}\label{sec-RE}

Given an input object $G$ and a category of rules, we are left  with the
problem of finding all relevant transformations of $G$ by these
rules. We cannot simply rely on the matchings of their left-hand sides
in $G$ (as in \cite{MaignanS15}) since they may not have pushout
complements, or they may have several non isomorphic ones. We will
therefore use the relevant direct transformations, albeit in an
abbreviated version that do not contain $L$, since we don't use
matchings, nor $H$ since they are not relevant to subsumption.

\begin{definition}[category $\partialcat$, functors $\inputF$, $\DPOmPF$]\label{def-partialcat}
  A \emph{partial transformation} $\ptobj$ in $\cat$ is a diagram
  \begin{center}
    \begin{tikzpicture}[xscale=1.5]
      \node (G) at (0,1) {$G$}; \node (D) at (1,1) {$D$};
      \node (K) at (2,1) {$K$};  \node (R) at (3,1) {$R$};  
      \path[->] (D) edge node[above, font=\footnotesize] {$f$} (G);
      \path[->] (K) edge node[above, font=\footnotesize] {$k$} (D);
      \path[->] (K) edge node[above, font=\footnotesize] {$r$} (R);
    \end{tikzpicture}
  \end{center}
  For any category $\cat$, let $\partialcat$ be the category whose
  objects are partial transformations and  morphisms
  $\ptmorph:\ptobj\rightarrow \ptobj'$ are triples
  $\tuple{\ptmorph_1, \ptmorph_2, \ptmorph_3}$ such that
  \begin{center}
    \begin{tikzpicture}[xscale=1.7,yscale=1.5]
      \node (G1) at (0,1) {$G$}; \node (D1) at (1,1) {$D$};
      \node (K1) at (2,1) {$K$};  \node (R1) at (3,1) {$R$};  
      \node (G2) at (0,0) {$G'$}; \node (D2) at (1,0) {$D'$};
      \node (K2) at (2,0) {$K'$};  \node (R2) at (3,0) {$R'$};  
      \path[->] (D1) edge node[above, font=\footnotesize] {$f$} (G1);
      \path[->] (K1) edge node[above, font=\footnotesize] {$k$} (D1);
      \path[->] (K1) edge node[above, font=\footnotesize] {$r$} (R1);
      \path[->] (D2) edge node[below, font=\footnotesize] {$f'$} (G2);
      \path[->] (K2) edge node[below, font=\footnotesize] {$k'$} (D2);
      \path[->] (K2) edge node[below, font=\footnotesize] {$r'$} (R2);
      \path[-] (G1) edge node[left, font=\footnotesize] {$=$} (G2);
      \path[->] (D2) edge node[left, font=\footnotesize] {$\ptmorph_1$} (D1);
      \path[->] (K1) edge node[left, font=\footnotesize] {$\ptmorph_2$} (K2);
      \path[->] (R1) edge node[right, font=\footnotesize] {$\ptmorph_3$} (R2);
    \end{tikzpicture}
  \end{center}
  commutes in $\cat$, with obvious composition and identities.

  Let $\inputF:\partialcat\rightarrow \discr{\cat}$ be the
  \emph{input functor} defined as $\inputF\ptobj = G$.
  Let $\DPOPF:\directcatDPO\rightarrow \partialcat$ and
  $\DPOmPF:\directcatDPOm\rightarrow \partialcat$ be the obvious
  functors (such that $(\DPOPF\directobj)\Gc = \directobj\Gc$ etc. and
  $\DPOPF\directmorph = \tuple{\directmorph_4, \directmorph_2, \directmorph_3}$).
\end{definition}

Using inverse images along $\DPOmPF$ and $\DPOmRF$ we can
easily focus on the direct transformations of concern (and the
morphisms between them), i.e., the transformations \emph{of} a graph
\emph{by} a rule. 

\begin{definition}[Rewriting Environments, rule systems, notations
  $\Dc_{\directobj}$, $\projd\directmorph$ \ldots]
  For any category $\cat$, a \emph{Rewriting Environment} for $\cat$
  consists of a category $\directcat$ of \emph{direct
    transformations}, a category $\rulecat$ of \emph{rules} and two
  functors
  \[\begin{tikzpicture}[xscale=1.5]
      \node(R) at (0,0){$\rulecat$}; \node(D) at (1,0){$\directcat$};
      \node(Cpt) at (2,0){$\partialcat$};
      \path[->] (D) edge node[above, font=\footnotesize] {$\ruleF$} (R);
      \path[->] (D) edge node[above, font=\footnotesize] {$\partialF$} (Cpt);
    \end{tikzpicture}\]

  A \emph{rule system} in a Rewriting Environment is a category
  $\rulesys$ with an embedding $\insertF:\rulesys\rightarrow \rulecat$
  (alternately, $\rulesys$ is a subcategory of $\rulecat$ and
  $\insertF$ is the inclusion functor).

  Given a rule system and an \emph{input} $\cat$-object $G$, we build the
  categories $\directcatG$, $\directcatSG$ and functors $\insertGF$,
  $\insertSF$, $\ruleFSG$ as meets of previous functors:
  \begin{center}
    \begin{tikzpicture}[xscale=1.7,yscale=1.5]
      \node(S) at (0,2){$\rulesys$}; \node(R) at (1,2){$\rulecat$};
      \node(D) at (1,1){$\directcat$}; \node(C) at (3,1){$\discr{\cat}$};
      \node(T) at (3,0){$\termcat$}; \node (P) at (2,1){$\partialcat$};
      \node(DG) at (1,0){$\directcatG$};
      \node(DSG) at (0,0){$\directcatSG$};
      \path[->] (D) edge node[left, font=\footnotesize] {$\ruleF$} (R);
      \path[>->] (S) edge node[above, font=\footnotesize] {$\insertF$} (R);
      \path[->] (D) edge node[above, font=\footnotesize] {$\partialF$} (P);
      \path[->] (P) edge node[above, font=\footnotesize] {$\inputF$} (C);
      \path[>->] (T) edge node[right, font=\footnotesize] {$G$} (C);
      \path[->] (DG) edge  (T);
      \path[>->] (DG) edge node[left, font=\footnotesize] {$\insertGF$} (D);
      \path[>->] (DSG) edge node[below, font=\footnotesize] {$\insertSF$} (DG);
      \path[->] (DSG) edge node[left, font=\footnotesize] {$\ruleFSG$} (S);
      \draw (0.3,0.1) to (0.3,0.3) to (0.1,0.3);
      \draw (1.3,0.1) to (1.3,0.3) to (1.1,0.3);
    \end{tikzpicture}
  \end{center}

  For any $\directobj\in\directcatSG$ we write $\Dc_{\directobj}$ for
  $(\partialF\insertGF\insertSF\directobj)\Dc$ and similarly
  $\fc_{\directobj}$ etc. For any
  $\directmorph:\directobj\rightarrow \directobj'$ in $\directcatSG$
  we write $\projd\directmorph$ for the first coordinate of
  $\partialF\insertGF\insertSF\directmorph$ and similarly
  $\projk\directmorph$, $\projr\directmorph$.
\end{definition}

\begin{example}
  For $\rulesys$ we take the subcategory
  \raisebox{-3.9ex}{\begin{tikzpicture}[xscale=1.5]
      \node (rho) at (0,0) {$\arule$}; \node (rho') at (1,0) {$\arule'$};
      \path[->,bend left] (rho) edge node[above, font=\footnotesize] {$\rulem^+$} (rho');
      \path[->,bend right] (rho) edge node[below, font=\footnotesize] {$\rulem^-$} (rho');
    \end{tikzpicture}} of $\rulecatDPO$. To the matchings $m'_1$ and
  $m'_2$ of $\arule'$ in $G$ correspond two\footnote{We consider
    transformations only up to isomorphisms, see Footnote~\ref{ft-2}.}
  transformations in $\directcatDPOm$ that will be denoted $\directobj'_1$ and
  $\directobj'_2$ (depicted on the right in
  Example~\ref{ex-pushouts}). To the matchings $ m^+$ and $m^-$ of
  $\arule$ in $G$ correspond another two transformations denoted
  $\directobj^+$ and $\directobj^-$ (on the left in
  Example~\ref{ex-pushouts}).  To each $i=1,2$ correspond one morphism
  $\directmorph^+_i: \directobj^+\rightarrow \directobj'_i$ such that
  $\DPOmRF \directmorph^+_i = \rulem^+$ and one morphism
  $\directmorph^-_i: \directobj^-\rightarrow \directobj'_i$ such that
  $\DPOmRF \directmorph_i^- = \rulem^-$.
  Thus $\directcatDPOm|_{G}^{\rulesys}$ is the following subcategory of $\directcatDPOm$. 
  \begin{center}
      \begin{tikzpicture}
    \node (T) at (0,1){$\directobj'_1$};
    \node (B) at (0,-1){$\directobj'_2$};
    \node (L) at (-1,0){$\directobj^+$};
    \node (R) at (1,0){$\directobj^-$};
    \path[->] (L) edge node[left, font=\footnotesize] {$\directmorph_1^+$} (T);
    \path[->] (R) edge node[right, font=\footnotesize] {\,$\directmorph_1^-$} (T);
    \path[->] (L) edge node[left, font=\footnotesize] {$\directmorph_2^+$} (B);
    \path[->] (R) edge node[right, font=\footnotesize] {$\directmorph_2^-$} (B);    
  \end{tikzpicture}
  \end{center}
\end{example}

\section{Global Coherent Transformations}\label{sec-GCT}
 
As stated above we will use the partial transformations that are
accessible from $\directcatSG$ through
$\partialF\circ \insertGF \circ \insertSF$ (a restriction of
$\partialF$). We first need to build a context between the input $G$
and the expected output $H$. In Parallel Coherent Transformation
\cite{BdlTE21a} the context is obtained as a limit of the morphisms
$\fc_{\directobj}:\Dc_{\directobj}\rightarrow G$ (that need not be
monics) for all $\directobj$ in a set $\stratcat$ of direct
transformations, hence of a diagram that is a sink to $G$ and thus
corresponds to a discrete diagram in $\slice{G}$. In Global Coherent
Transformations the \emph{global context} (denoted $\interf$ below) is
obtained similarly, but now $\stratcat$ is a category and the diagram
contains the morphisms
$\projd\directmorph: \fc_{\directobj'} \rightarrow \fc_{\directobj}$
for all $\directmorph: \directobj \rightarrow \directobj'$ in
$\stratcat$ (since
$\fc_{\directobj}\circ \projd\directmorph= \fc_{\directobj'}$).

\begin{definition}[functor $\ProjG$, limit
  $f_{\stratcat}:\interf\rightarrow G$, limit cone $\intercone$]
  For any subcategory $\stratcat$ of $\directcatSG$ let
  $\ProjG:\dual{\stratcat}\rightarrow \slice{G}$ be the contravariant
  functor that maps every $\directobj\in\stratcat$ to
  $\fc_{\directobj}: \Dc_{\directobj} \rightarrow G$ and every
  morphism $\directmorph$ of $\stratcat$ to
  $\projd\directmorph: \fc_{\directobj'} \rightarrow
  \fc_{\directobj}$.  Let $f_{\stratcat}:\interf\rightarrow G$ be the
  limit of $\ProjG$ and $\intercone$ be the limit cone from
  $f_{\stratcat}$ to $\ProjG$.
\end{definition}

Note that if $\stratcat$ is empty then the limit $f_{\stratcat}$ of
the empty diagram is the terminal object of $\slice{G}$, that is
$\id{G}$, hence $\interf=G$. 

\begin{example}
  Let $\stratcat=\directcatDPOm|_{G}^{\rulesys}$. The diagram on the
  left below corresponds to the functor $\ProjG$ together with the
  morphisms
  $\fc_{\directobj^{\pm}_i}: \Dc_{\directobj^{\pm}_i}\rightarrow G$
  (objects in $\slice{G}$). The limit of this diagram yields
  $\interf = \dvgr$ and the limit cone is represented on the right.

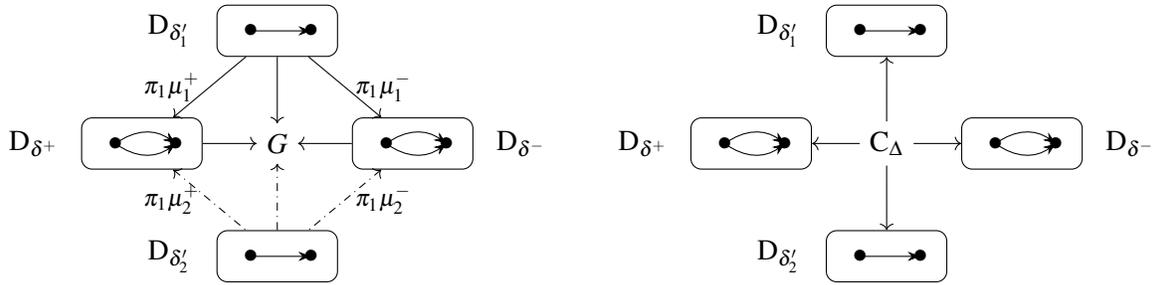
\begin{center}
  \begin{tikzpicture}[yscale=1.5,xscale=1.8]
    \node[draw,rounded corners] (T1) at (0,1){\egr};
    \node[draw,rounded corners] (B1) at (0,-1){\egr};
    \node[draw,rounded corners] (L1) at (-1,0){\degr};
    \node[draw,rounded corners] (R1) at (1,0){\degr};
    \node (C1) at (0,0){$G$};
    \node at (-0.8,1){$\Dc_{\directobj'_1}$}; 
    \node at (-0.8,-1){$\Dc_{\directobj'_2}$}; 
    \node at (-1.8,0){$\Dc_{\directobj^+}$}; 
    \node at (1.8,0){$\Dc_{\directobj^-}$}; 
    \path[<-] (L1) edge node[left, font=\footnotesize] {$\projd\directmorph_1^+$} (T1);
    \path[<-] (R1) edge node[right, font=\footnotesize] {$\projd\directmorph_1^-$} (T1);
    \path[<-,dash dot] (L1) edge node[left, font=\footnotesize] {$\projd\directmorph_2^+$} (B1);
    \path[<-,dash dot] (R1) edge node[right, font=\footnotesize] {$\projd\directmorph_2^-$} (B1);    
    \path[->](L1) edge (C1);
    \path[->](R1) edge (C1);
    \path[->](T1) edge (C1);
    \path[->,dash dot](B1) edge (C1);
    \node[draw,rounded corners] (T2) at (4.5,1){\egr};
    \node[draw,rounded corners] (B2) at (4.5,-1){\egr};
    \node[draw,rounded corners] (L2) at (3.5,0){\degr};
    \node[draw,rounded corners] (R2) at (5.5,0){\degr};
   \node (C2) at (4.5,0){$\interf$};
    \node at (4.5-0.8,1){$\Dc_{\directobj'_1}$}; 
    \node at (4.5-0.8,-1){$\Dc_{\directobj'_2}$}; 
    \node at (4.5-1.8,0){$\Dc_{\directobj^+}$}; 
    \node at (4.5+1.8,0){$\Dc_{\directobj^-}$}; 
   \path[<-](L2) edge (C2);
   \path[<-](R2) edge (C2);
   \path[<-](T2) edge (C2);
   \path[<-](B2) edge (C2);
  \end{tikzpicture}
\end{center}
\end{example}

We next need to check that the transformations in $\stratcat$ do not
conflict with each other, i.e., that for all $\directobj\in\stratcat$
the image of $\Kc_{\directobj}$ in $G$ is not only preserved by
$\directobj$ (in $\Dc_{\directobj}$) but also by all other
transformations $\directobj'\in\stratcat$. This is ensured by finding
(natural) cones from these $\Kc_{\directobj}$ to the
$\Dc_{\directobj'}$, which we shall formulate with $\ProjG$, hence in
$\slice{G}$.

\begin{definition}[coherent system of cones, morphisms
  $c_{\directobj}$, global coherence]
  A \emph{coherent system of cones for} $\stratcat$ is a set of cones
  $\cone{\directobj}$ from $\fc_{\directobj}\circ \kc_{\directobj}$ to
  $\ProjG$ such that $\cone{\directobj}\directobj = \kc_{\directobj}$
  for all $\directobj\in \stratcat$, and
  $\cone{\directobj} = \cone{\directobj'}\circ \projk\directmorph$ for
  all $\directmorph:\directobj\rightarrow\directobj'$ in $\stratcat$.
  $\stratcat$ is \emph{globally coherent} if there exists a coherent
  system of cones for $\stratcat$.  We then let
  $c_{\directobj}: \fc_{\directobj}\circ \kc_{\directobj} \rightarrow
  f_{\stratcat}$ be the unique morphism in $\slice{G}$ such that
  $\cone{\directobj} = \intercone \circ c_{\directobj}$.
\end{definition}

Note that if $\cone{\directobj'}$ is a cone from
$\fc_{\directobj'}\circ \kc_{\directobj'}$ to $\ProjG$ then
$\cone{\directobj'}\circ \projk\directmorph$ is a cone from
$\fc_{\directobj}\circ \kc_{\directobj}$ to $\ProjG$, hence global
coherence means that we should find cones for overlapping direct
transformations (say $\directobj'_1$ and $\directobj'_2$), with the
constraint that they should be compatible on their common
subtransformations
$\directobj'_1\leftarrow\directobj \rightarrow \directobj'_2$. If
$\rulesys$ and therefore $\stratcat$ are discrete, this amounts to
parallel coherence (that generalizes parallel independence in DPO, see
\cite{BdlTE21a}).

\begin{example}
  On our example the four graphs $\Kc_{\directobj^{\pm}_i}$ are equal to
$\dvgr$. It is easy to build the four cones from the four
morphisms from $\Kc_{\directobj'_i}$ to $\Dc_{\directobj'_i}$ depicted below,
by composing them with the $\projd\directmorph^{\pm}_i$ on the left
and the $\projk\directmorph^{\pm}_i$ on the right. On the right are
also depicted the morphisms $c_{\directobj^{\pm}_i}$.

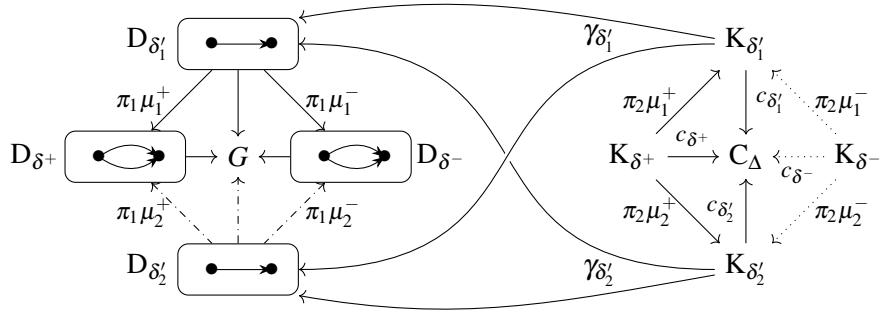
\begin{center}
  \begin{tikzpicture}[scale=1.5]
    \node[draw,rounded corners] (T1) at (0,1){\egr};
    \node[draw,rounded corners] (B1) at (0,-1){\egr};
    \node[draw,rounded corners] (L1) at (-1,0){\degr};
    \node[draw,rounded corners] (R1) at (1,0){\degr};
   \node (C1) at (0,0){$G$};
    \node at (-0.8,1){$\Dc_{\directobj'_1}$}; 
    \node at (-0.8,-1){$\Dc_{\directobj'_2}$}; 
    \node at (-1.8,0){$\Dc_{\directobj^+}$}; 
    \node at (1.8,0){$\Dc_{\directobj^-}$}; 
    \node (T2) at (4.5,1){$\Kc_{\directobj'_1}$};
    \node (B2) at (4.5,-1){$\Kc_{\directobj'_2}$};
    \node (L2) at (3.5,0){$\Kc_{\directobj^+}$};
    \node (R2) at (5.5,0){$\Kc_{\directobj^-}$};
    \node (C2) at (4.5,0){$\interf$};
   \path[->,out=180,in=0,looseness=1.5](B2) edge (T1);
   \path[->,out=180, in=0,looseness=1.5](T2) edge[draw=white, line width = 3pt] (B1); 
   \path[->,out=180, in=0,looseness=1.5](T2) edge (B1); 
   \path[<-] (L1) edge node[left, font=\footnotesize] {$\projd\directmorph_1^+$} (T1);
    \path[<-] (R1) edge node[right, font=\footnotesize] {$\projd\directmorph_1^-$} (T1);
    \path[<-,dash dot] (L1) edge node[left, font=\footnotesize] {$\projd\directmorph_2^+$} (B1);
    \path[<-, dash dot] (R1) edge node[right, font=\footnotesize] {$\projd\directmorph_2^-$} (B1);    
    \path[->](L1) edge (C1);
    \path[->](R1) edge (C1);
    \path[->](T1) edge (C1);
    \path[->, dash dot](B1) edge (C1);
    \path[->] (L2) edge node[left, font=\footnotesize] {$\projk\directmorph_1^+$} (T2);
    \path[->,dotted] (R2) edge node[right, font=\footnotesize] {$\projk\directmorph_1^-$} (T2);
    \path[->] (L2) edge node[left, font=\footnotesize] {$\projk\directmorph_2^+$} (B2);
    \path[->,dotted] (R2) edge node[right, font=\footnotesize] {$\projk\directmorph_2^-$} (B2);    
    \path[->](L2) edge node[above, font=\footnotesize] {$c_{\directobj^+}$} (C2);
    \path[->,dotted](R2) edge node[below, font=\footnotesize] {$c_{\directobj^-}$} (C2);
    \path[->](T2) edge node[right, font=\footnotesize] {$c_{\directobj'_1}$} (C2);
    \path[->](B2) edge node[left, font=\footnotesize] {$c_{\directobj'_2}$} (C2);
    \path[->,out=170, in=10](T2) edge (T1.north east);
    \path[->,out=190,in=-10](B2) edge (B1.south east);
    \node at (3.2,1.05){$\cone{\directobj'_1}$}; 
    \node at (3.2,-1.03){$\cone{\directobj'_2}$}; 
  \end{tikzpicture}
\end{center}
The reader may check that
$\cone{\directobj'_1}\circ \projk\directmorph_1^+ =
\cone{\directobj'_2}\circ \projk\directmorph_2^+$ (this is
$\cone{\directobj^+}$) and
$\cone{\directobj'_1}\circ \projk\directmorph_1^- =
\cone{\directobj'_2}\circ \projk\directmorph_2^-$
($= \cone{\directobj^-}$).
\end{example}

The morphisms $c_{\directobj}$ specify where the right-hand sides
$\Rc_{\directobj}$ should be pushed in the global context.

\begin{definition}[morphisms $h_{\directobj}:\interf\rightarrow H_{\directobj}$]
  If $\stratcat$ is globally coherent for all $\directobj\in\stratcat$
  then $c_{\directobj}$ can be viewed as a $\cat$-morphism
  $c_{\directobj}: \Kc_{\directobj} \rightarrow \interf$, and we
  consider the following pushout in $\cat$.
  \begin{center}
    \begin{tikzpicture}[xscale=1.5,yscale=1.5]
      \node(K) at (0,1){$\Kc_{\directobj}$};
      \node(C) at (0,0){$\interf$}; \node(H) at (1,0){$H_{\directobj}$};
      \node(R) at (1,1){$\Rc_{\directobj}$};
      \path[->] (K) edge node[above, font=\footnotesize] {$\rc_{\directobj}$} (R);
      \path[->] (K) edge node[left, font=\footnotesize] {$c_{\directobj}$} (C);
      \path[->] (R) edge node[right, font=\footnotesize] {$n_{\directobj}$} (H);
      \path[->] (C) edge node[below, font=\footnotesize] {$h_{\directobj}$} (H);
      \draw (0.7,0.1) to (0.7,0.3) to (0.9,0.3);
    \end{tikzpicture}
  \end{center}
\end{definition}

\begin{example}
  On our example we get:
  
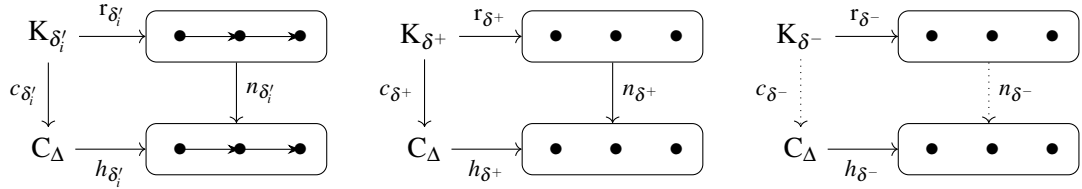
\begin{center}
    \begin{tikzpicture}[xscale=2.5,yscale=1.5]
      \node(K1) at (0,1){$\Kc_{\directobj'_i}$};
      \node(C1) at (0,0){$\interf$};
      \node[draw,rounded corners](H1) at (1,0){$\tvdegr$};
      \node[draw,rounded corners](R1) at (1,1){$\tvdegr$};
      \path[->] (K1) edge node[above, font=\footnotesize] {$\rc_{\directobj'_i}$}  (R1);
      \path[->] (K1) edge node[left, font=\footnotesize] {$c_{\directobj'_i}$} (C1);
      \path[->] (R1) edge node[right, font=\footnotesize] {$n_{\directobj'_i}$} (H1);
      \path[->] (C1) edge node[below, font=\footnotesize] {$h_{\directobj'_i}$} (H1);
      \node(K2) at (2,1){$\Kc_{\directobj^+}$};
      \node(C2) at (2,0){$\interf$};
      \node[draw,rounded corners](H2) at (3,0){$\tvgr$};
      \node[draw,rounded corners](R2) at (3,1){$\tvgr$};
      \path[->] (K2) edge node[above, font=\footnotesize] {$\rc_{\directobj^+}$} (R2);
      \path[->] (K2) edge node[left, font=\footnotesize] {$c_{\directobj^+}$} (C2);
      \path[->] (R2) edge node[right, font=\footnotesize] {$n_{\directobj^+}$} (H2);
      \path[->] (C2) edge node[below, font=\footnotesize] {$h_{\directobj^+}$} (H2);
      \node(K3) at (4,1){$\Kc_{\directobj^-}$};
      \node(C3) at (4,0){$\interf$};
      \node[draw,rounded corners](H3) at (5,0){$\tvgr$};
      \node[draw,rounded corners](R3) at (5,1){$\tvgr$};
      \path[->] (K3) edge node[above, font=\footnotesize] {$\rc_{\directobj^-}$} (R3);
      \path[->,dotted] (K3) edge node[left, font=\footnotesize] {$c_{\directobj^-}$} (C3);
      \path[->,dotted] (R3) edge node[right, font=\footnotesize] {$n_{\directobj^-}$} (H3);
      \path[->] (C3) edge node[below, font=\footnotesize] {$h_{\directobj^-}$} (H3);
    \end{tikzpicture}
  \end{center}
\end{example}

Thanks to the coherent system of cones we can turn $h$ into a functor.

\begin{proposition}\label{pr-h}
  For every $\directmorph:\directobj\rightarrow \directobj'$ in
  $\stratcat$ there exists a unique $h_{\directmorph}$ such that
  \begin{center}
    \begin{tikzpicture}[xscale=2.5,yscale=0.8]
      \node(C) at (0,0){$\interf$};
      \node(H) at (1,1){$H_{\directobj}$};
      \node(H') at (1,-1){$H_{\directobj'}$};
      \node(R) at (2,1){$\Rc_{\directobj}$};
      \node(R') at (2,-1){$\Rc_{\directobj'}$};
      \path[->] (C) edge node[above, font=\footnotesize] {$h_{\directobj}$} (H);
      \path[->] (C) edge node[below, font=\footnotesize] {$h_{\directobj'}$} (H');
      \path[->] (R) edge node[above, font=\footnotesize] {$n_{\directobj}$} (H);
      \path[->] (R') edge node[below, font=\footnotesize] {$n_{\directobj'}$} (H');
      \path[->] (R) edge node[right, font=\footnotesize] {$\projr\directmorph$} (R');
      \path[->,dashed] (H) edge node[right, font=\footnotesize] {$h_{\directmorph}$} (H');
    \end{tikzpicture}
  \end{center}
  commutes.
\end{proposition}

\begin{corollary}
  By unicity we get $h_{\directmorph'\circ\directmorph} = h_{\directmorph'} \circ h_{\directmorph}$.
\end{corollary}

\begin{example}
  For instance the morphisms $\directmorph_i^- : \directobj^-
\rightarrow \directobj'_i$ yield the morphisms $h_{\directmorph_i^-}$
depicted below.

  
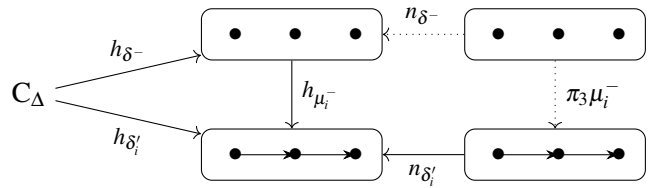
\begin{center}
    \begin{tikzpicture}[xscale=3.5,yscale=0.8]
      \node(C1) at (0,0){$\interf$};
      \node[draw,rounded corners](H) at (1,1){\tvgr};
      \node[draw,rounded corners](H') at (1,-1){\tvdegr};
      \node[draw,rounded corners](R) at (2,1){\tvgr};
      \node[draw,rounded corners](R') at (2,-1){\tvdegr};
      \path[->] (C1) edge node[above, font=\footnotesize] {$h_{\directobj^-}$} (H);
      \path[->] (C1) edge node[below, font=\footnotesize] {$h_{\directobj'_i}$} (H');
      \path[->,dotted] (R) edge node[above, font=\footnotesize] {$n_{\directobj^-}$} (H);
      \path[->] (R') edge node[below, font=\footnotesize] {$n_{\directobj'_i}$} (H');
      \path[->,dotted] (R) edge node[right, font=\footnotesize] {$\projr\directmorph^-_i$} (R');
      \path[->] (H) edge node[right, font=\footnotesize] {$h_{\directmorph^-_i}$} (H');
    \end{tikzpicture}
  \end{center}
\end{example}

The final step of the Global Coherent Transformation, symmetric to
the first step, consists in taking the colimit in the coslice category
$\coslice{\interf}$ of the covariant diagram of index $\stratcat$ with
objects $h_{\directobj}$ and morphisms
$h_{\directmorph}: h_{\directobj}\rightarrow h_{\directobj'}$ for all
$\directmorph: \directobj\rightarrow \directobj'$ in $\stratcat$.

\begin{definition}[functor $\ProjH$, colimit $h_{\stratcat}:\interf\rightarrow H_{\stratcat}$]
  If $\stratcat$ is globally coherent let
  $\ProjH:\stratcat\rightarrow \coslice{\interf}$ be the functor
  defined by $\ProjH \directobj=h_{\directobj}$ (interpreted as an
  object of $\coslice{\interf}$) and
  $\ProjH \directmorph=h_{\directmorph}$ for all
  $\directmorph:\directobj\rightarrow \directobj'$ in $\stratcat$. Let
  $h_{\stratcat}:\interf\rightarrow H_{\stratcat}$ be the
  colimit\footnote{\label{ft-2} Global Coherent Transformations are obtained as
    limits and colimits of diagrams whose index category is
    $\stratcat$, hence are not affected by isomorphisms in $\stratcat$,
    which can therefore be replaced by its skeleton.} of
  $\ProjH$, then the $\cat$-span
  $G\xleftarrow{f_{\stratcat}} \interf \xrightarrow{h_{\stratcat}}
  H_{\stratcat}$ is a \emph{Global Coherent Transformation by
    $\stratcat$}.
\end{definition}

If $\stratcat$ is empty then the colimit $h_{\stratcat}$ of the empty
diagram is the initial object of $\coslice{\interf}$, that is
$\id{\interf}$, hence $H_{\stratcat} = \interf = G$. Generally, the
functor $\ProjH$ depends on the choice of cones $\cone{\directobj}$
for $\directobj\in\stratcat$, hence $h_{\stratcat}$ is not determined
by $\stratcat$.

\begin{example}\label{ex-result}
  The functor $\ProjH$ applied to $\stratcat$ yields the following
  diagram

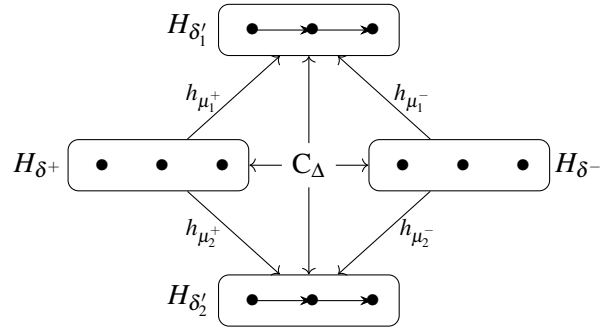
\begin{center}
  \begin{tikzpicture}[xscale=2,yscale=1.8]
    \node[draw,rounded corners] (T2) at (3.5,1){\tvdegr};
    \node[draw,rounded corners] (B2) at (3.5,-1){\tvdegr};
    \node[draw,rounded corners] (L2) at (2.5,0){\tvgr};
    \node[draw,rounded corners] (R2) at (4.5,0){\tvgr};
    \node (C2) at (3.5,0){$\interf$};
    \node at (2.7,1){$H_{\directobj'_1}$}; 
    \node at (2.7,-1){$H_{\directobj'_2}$}; 
    \node at (1.7,0){$H_{\directobj^+}$}; 
    \node at (5.3,0){$H_{\directobj^-}$}; 

    \path[->] (L2) edge node[left, font=\footnotesize] {$h_{\directmorph_1^+}$} (T2);
    \path[->] (R2) edge node[right, font=\footnotesize] {$h_{\directmorph_1^-}$} (T2);
    \path[->] (L2) edge node[left, font=\footnotesize] {$h_{\directmorph_2^+}$} (B2);
    \path[->] (R2) edge node[right, font=\footnotesize] {$\,\,h_{\directmorph_2^-}$} (B2);    
    \path[<-](L2) edge (C2);
    \path[<-](R2) edge (C2);
    \path[<-](T2) edge (C2);
    \path[<-](B2) edge (C2);
  \end{tikzpicture}
\end{center}
The leftmost vertices of these five graphs are connected as images or
preimages of each other, and similarly for the five right vertices,
and the four middle vertices. The four edges are not likewise
connected, hence the colimit of this diagram is the expected result
$H = \tvqegr$. We therefore see that the two middle vertices created
in $\directobj'_1$ and $\directobj'_2$ are merged by their common
subtransformation $\directobj^+$ (or $\directobj^-$), but also that
the two middle vertices created in $\directobj^+$ and $\directobj^-$
are merged by their common subsuming transformation $\directobj'_1$ (or
$\directobj'_2$).

If  we apply $\rulesys$ to the graph $G'= \dvgr$ then
rule $\arule'$ does not apply to $G'$ and hence the two matchings of
$\arule$ in $G'$ apply independently, thus adding two vertices to
$G'$. We can merge them by adding to $\rulesys$ the following rule
morphism $\rulem:\arule\rightarrow\arule$ that swaps the left and
right vertices:

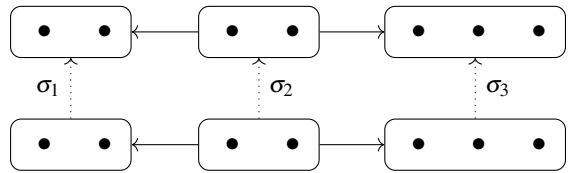
\begin{center}
  \begin{tikzpicture}[xscale=2.5,yscale=1.5]
    \node[draw,rounded corners](L) at (0,0) {$\dvgr$};
    \node[draw,rounded corners](K) at (1,0){$\dvgr$};
    \node[draw,rounded corners](R) at (2.15,0){$\tvgr$};
    \path[->](K) edge (L); \path[->](K) edge (R);
    \node[draw,rounded corners](L') at (0,1) {$\dvgr$};
    \node[draw,rounded corners](K') at (1,1){$\dvgr$};
    \node[draw,rounded corners](R') at (2.15,1){$\tvgr$};
    \path[->](K') edge (L'); \path[->](K') edge (R');
    \path[->,dotted](L) edge node[left, font=\footnotesize] {$\rulem_1$} (L'); 
    \path[->,dotted](K) edge node[right, font=\footnotesize] {$\rulem_2$} (K'); 
    \path[->,dotted](R) edge node[right, font=\footnotesize] {$\rulem_3$} (R'); 
  \end{tikzpicture}
\end{center}

We have $\rulem^2 = \id{\arule}$ hence $\rulem$ is an automorphism of
$\arule$. Adding $\rulem$ to $\rulesys$ means that the symmetric
applications of $\arule$, i.e., direct transformations with matchings
$m$ and $m\circ\rulem$, shall be merged (this seems to generalize to
the algebraic context the notion of Parallel Rewriting Modulo
Automorphism devised in an algorithmic approach in \cite{BdlTE20c}).
Since $\rulem^+\circ\rulem = \rulem^-$ and
$\rulem^-\circ\rulem = \rulem^+$, the new rule system is
\[\rulesys' = \raisebox{-3.9ex}{\begin{tikzpicture}[xscale=1.5] \node
    (rho) at (0,0) {$\arule$}; \node (rho') at (1,0)
    {$\arule'$}; \path[->,bend left] (rho) edge node[above,
    font=\footnotesize]
    {$\rulem^+$} (rho'); \path[->,bend right] (rho) edge node[below,
    font=\footnotesize]
    {$\rulem^-$} (rho'); \path[->,out=-135,in=135,looseness=2.5] (rho)
     edge node[left, font=\footnotesize] {$\rulem$} (rho);
\end{tikzpicture}}\] If we apply $\rulesys'$ to $G$,  we add two new morphisms in
$\directcatDPOm|_{G}^{\rulesys}$, i.e, \[\stratcat'=\directcatDPOm|_{G}^{\rulesys'}=
\raisebox{-6.9ex}{
  \begin{tikzpicture}
    \node (T) at (0,1){$\directobj'_1$};
    \node (B) at (0,-1){$\directobj'_2$};
    \node (L) at (-1,0){$\directobj^+$};
    \node (R) at (1,0){$\directobj^-$};
    \path[->] (L) edge node[left, font=\footnotesize] {$\directmorph_1^+$} (T);
    \path[->] (R) edge node[right, font=\footnotesize] {$\,\directmorph_1^-$} (T);
    \path[->] (L) edge node[left, font=\footnotesize] {$\directmorph_2^+$} (B);
    \path[->] (R) edge node[right, font=\footnotesize]
    {$\directmorph_2^-$} (B);
    \path[->,bend left] (L) edge (R); \path[->,bend left] (R) edge (L);
  \end{tikzpicture}}\]
It is easy to see that the Global Coherent Transformation by
$\stratcat'$ is the same as above with $\stratcat$. This is due to the
fact that $\directobj^+$ and $\directobj^-$ are already related in
$\stratcat$ through $\directobj'_1$ (or $\directobj'_2$).
\end{example}

We finally prove that, apart from this mechanism of sharing common
subtransformations, isolated transformations always subsume their
subtransformations, so that morphisms in $\rulecat$ are rule
subsumptions as intended.

\begin{proposition}\label{pr-subsume}
  If $\stratcat'$ is restricted to $\directobj'$ and $\stratcat$ to
  $\directmorph:\directobj\rightarrow \directobj'$ (or more generally
  if $\directobj'$ is terminal in $\stratcat$) then $\stratcat$ and
  $\stratcat'$ are globally coherent and
  $H_{\stratcat} \simeq H_{\stratcat'}$.
\end{proposition}

\section{Some Rewriting Environments and Their Properties}\label{sec-env}

An obvious property of Rewriting Environments is that they can be
combined: if
$\rulecat_1\xleftarrow{\ruleF_1}\directcat_1\xrightarrow{\partialF_1}
\partialcat$ and
$\rulecat_2\xleftarrow{\ruleF_2}\directcat_2\xrightarrow{\partialF_2}
\partialcat$ are Rewriting Environments for $\cat$ then so is
$\rulecat_1+\rulecat_2 \xleftarrow{\ruleF_1+\ruleF_2}
\directcat_1+\directcat_2 \xrightarrow{[\partialF_1,\partialF_2]}
\partialcat$. It is therefore possible to mix rules of different
approaches to transform a graph, though of course rules of distinct
approaches cannot subsume each other.

A property that one might reasonably expect is that when a rule
applies and yields a direct transformation then its subrules also
apply and yield subtransformations.  We express this by means of the
following notion.

\begin{definition}[right-full]
  A functor $\aFunc:\mathcal{A}\rightarrow\mathcal{B}$ is
  \emph{right-full}\footnote{This is named after the symmetric
    definition of \emph{left-full} functors in
    \cite[p. 63]{SellingerPhd97}.} if for all $a'\in\mathcal{A}$, all
  $b\in\mathcal{B}$ and all $g:b\rightarrow Fa'$, there exist
  $a\in\mathcal{A}$ and $f:a\rightarrow a'$ such that $Ff=g$.
\end{definition}
It is obvious that right-fullness is closed by composition.

\begin{lemma}\label{lm-IGF-full}
  $\insertGF$ is a full and right-full embedding.
\end{lemma}

\begin{proposition}\label{pr-right-full}
  If $\ruleF$ is right-full (resp. faithful) then so is $\ruleFSG$ for
  every rule system $\rulesys$ and 
$G\in\cat$. 
\end{proposition}

Hence when $\ruleF$ is right-full and faithful every morphism
$\rulem:\arule\rightarrow\arule'$ in $\rulesys$ is reflected by a
morphism in $\directcatSG$ whenever $\arule'$ is reflected by a direct
transformation $\directobj'$ (i.e., whenever $\arule'$ applies to
$G$), and this morphism is uniquely determined by $\rulem$ and
$\directobj'$. 

\subsection{Double-Pushouts}

Definitions \ref{def-catDPO}, \ref{def-transfoDPO} and
\ref{def-partialcat} provide two Rewriting Environments that we may
call DPO and DPOm. By Proposition~\ref{prop-DPO} it is obvious that
$\DPORF$ and $\DPOmRF$ are faithful  when $\cat$ is the category of
graphs. This is easily seen to generalize to all adhesive categories
\cite{LackS05}. Proposition~\ref{prop-DPO} generalizes
as follows:

\begin{proposition}\label{pr-adhesive}
  If $\cat$ is adhesive, $\directobj,\directobj'\in\directcatDPO$ and
  $\rulem: \DPORF\directobj \rightarrow \DPORF\directobj'$ such that
  $m=m' \circ \rulem_1$ then there exists a unique
  $\directmorph:\directobj\rightarrow \directobj'$ such that
  $\DPORF\directmorph = \rulem$.
\end{proposition}

According to Proposition~\ref{prop-gc} it is obvious
that $\DPOmRF$ is right-full (when $\cat$ is the category of
graphs). It is easy to see that $\DPORF$ is not right-full 
(with $\rulem_1$ not monic, see Proposition~\ref{prop-gc}).

One drawback with span rules is that every item matched by $m$ that is
not removed must be preserved in the result, hence cannot be removed
by an overlapping rule, by the requirement of global coherence. In
\cite{BdlTE21a} we have defined \emph{weak} DPO-rules by inserting a
second interface $I$ between $K$ and $L$. A weak
DPO transformation is a diagram
\begin{center}
  \begin{tikzpicture}[xscale=1.8, yscale=1.5]
    \node (L) at (0,1) {$L$}; \node (K) at (1,1) {$I$};  \node (I) at
    (2,1) {$K$};  \node (R) at (3,1) {$R$}; \node (G) at (0,0) {$G$};
    \node (D1) at (1,0) {$D$}; \node (D2) at (2,0) {$D$}; \node (H) at (3,0) {$H$};
    \path[>->] (K) edge  node[above,fill=white, font=\footnotesize] {$l$} (L);
    \path[>->] (I) edge  node[above,fill=white, font=\footnotesize] {$i$} (K);
    \path[->] (I) edge  node[above,fill=white, font=\footnotesize] {$r$} (R);
    \path[->] (L) edge node[left,fill=white, font=\footnotesize] {$m$} (G);
    \path[->] (K) edge node[left,fill=white, font=\footnotesize] {$k$} (D1);
    \path[->] (D1) edge node[below,fill=white, font=\footnotesize] {$f$} (G);
    \path[->] (I) edge  node[left,fill=white, font=\footnotesize] {$k\circ i$} (D2);
    \path[->] (D2) edge node[below,fill=white, font=\footnotesize] {$g$} (H);
    \path[->] (R) edge  node[right,fill=white, font=\footnotesize] {$n$} (H);
    \path[-] (D1) edge node[above,fill=white, font=\footnotesize] {$=$} (D2);
    \draw (0.3,0.1) to (0.3,0.3) to (0.1,0.3);
    \draw (2.7,0.1) to (2.7,0.3) to (2.9,0.3);
  \end{tikzpicture}
\end{center}
so that the images of items in $I$ are not removed by this
transformation, but only images of items in $K$ may not be removed by
any simultaneous transformation. In cellular automata we need items in
$I$ that match the states of the neighbour cells, but there should be
none in $K$ since these states may be modified by overlapping rules
(see \cite[Example 3]{BdlTE21a}, note that $K$ and $I$ are swapped).

It is easy to define subsumption morphisms between weak DPO-rules (as
4-tuples of $\cat$-morphisms with commuting properties and a pullback
as in Definition \ref{def-catDPO}), and corresponding morphisms
between direct transformations of weak DPO-rules (as 5-tuples of
$\cat$-morphisms with commuting properties and a pullback as in
Definition \ref{def-transfoDPO}). This yields a Rewriting Environment
for weak double-pushouts.

\subsection{Sesqui-Pushouts}

We now consider the case of Sesqui-Pushouts \cite{CorradiniHHK06}. It
is based on the notion of final pullback complement that allows not
only to remove parts of the input $G$ but also to make copies of parts
of $G$. 

\begin{definition}[category $\rulecatSqPO$, direct SqPO-transformations]\label{def-SqPO}
  A \emph{SqPO-rule} $\arule$ in $\cat$ is a span diagram
  $L\xleftarrow{l} K \xrightarrow{r}R$.
  Let $\rulecatSqPO$ be the category whose
  objects are the SqPO-rules and morphisms
  $\rulem:\arule\rightarrow\arule'$ are triples
  $\rulem=\tuple{\rulem_1,\rulem_2,\rulem_3}$ such that
  \begin{center}
    \begin{tikzpicture}[xscale=1.7,yscale=1.5]
      \node (D1) at (1,1) {$L$};
      \node (K1) at (2,1) {$K$};  \node (R1) at (3,1) {$R$};  
      \node (D2) at (1,0) {$L'$};
      \node (K2) at (2,0) {$K'$};  \node (R2) at (3,0) {$R'$};  
      \path[->] (K1) edge node[above, font=\footnotesize] {$l$} (D1);
      \path[->] (K1) edge node[above, font=\footnotesize] {$r$} (R1);
      \path[->] (K2) edge node[below, font=\footnotesize] {$l'$} (D2);
      \path[->] (K2) edge node[below, font=\footnotesize] {$r'$} (R2);
      \path[->] (D1) edge node[left, font=\footnotesize] {$\rulem_1$} (D2);
      \path[->] (K1) edge node[right, font=\footnotesize] {$\rulem_2$} (K2);
      \path[->] (R1) edge node[right, font=\footnotesize] {$\rulem_3$} (R2);
      \draw (1.7,0.9) to (1.7,0.7) to (1.9,0.7);
    \end{tikzpicture}
  \end{center}
  commutes in $\cat$ and the left square is a pullback, with obvious
  composition and identities. Let $\rulecatSqPOm$ be the subcategory with
  morphisms $\rulem$ such that $\rulem_1$ and $\rulem_2$ are monics.

  A \emph{final pullback complement} of $\tuple{m,l}$ is a pair
  $\tuple{f,k}$ such that $\tuple{k,l}$ is a pullback of $\tuple{f,m}$
  and for every pullback $\tuple{k',l\circ c}$ of any $\tuple{f',m}$
  there exists a unique $d$ such that
\begin{center}
  \begin{tikzpicture}[xscale=1.9,yscale=1.6,z=-0.7cm]
    \node(G) at (0,0,0){$G$}; \node(G') at (0,0,1){$G$};
    \node(L) at (0,1,0){$L$};
    \node(L') at (0,1,1){$L$};
    \node(K) at (1,1,0){$K$};
    \node(K') at (1,1,1){$K'$};
    \node(D) at (1,0,0){$D$};
    \node(D') at (1,0,1){$D'$};
    \path[-](G) edge node[above, font=\footnotesize] {=}(G');
    \path[-](L) edge node[above, font=\footnotesize] {$=$}(L');
    \path[<-](K) edge node[right, font=\footnotesize] {$c$}(K');
    \path[<-,dashed](D) edge node[right, font=\footnotesize] {$d$}(D');
    \path[<-](L) edge node[above, font=\footnotesize] {$l$}(K);
    \path[->](L) edge node[right, near start, font=\footnotesize] {$m$}(G);
    \path[->](L') edge node[left, font=\footnotesize] {$m$}(G');
    \path[<-](G) edge node[above,near end, font=\footnotesize] {$f$}(D);
    \path[<-](G') edge node[below, font=\footnotesize] {$f'$}(D');
    \path[->](K) edge node[right, font=\footnotesize] {$k$}(D);
    \path[-](L') edge[draw=white, line width=3pt] (K');
    \path[<-](L') edge node[above, font=\footnotesize] {$l\circ c$}(K');
    \path[-](K') edge[draw=white, line width=3pt] (D');
    \path[->](K') edge node[left,near end, font=\footnotesize] {$k'$}(D');
  \end{tikzpicture}
\end{center}
commutes.

  A \emph{direct SqPO-transformation} in $\cat$ is a diagram
  \begin{center}
    \begin{tikzpicture}[xscale=1.7,yscale=1.5]
      \node (D1) at (0,1) {$L$};
      \node (K1) at (1,1) {$K$};  \node (R1) at (2,1) {$R$};  
      \node (D2) at (0,0) {$G$};
      \node (K2) at (1,0) {$D$};  \node (R2) at (2,0) {$H$};  
      \path[->] (K1) edge node[above, font=\footnotesize] {$l$} (D1);
      \path[->] (K1) edge node[above, font=\footnotesize] {$r$} (R1);
      \path[->] (K2) edge node[below, font=\footnotesize] {$f$} (D2);
      \path[->] (K2) edge node[below, font=\footnotesize] {$g$} (R2);
      \path[->] (D1) edge node[left, font=\footnotesize] {$m$} (D2);
      \path[->] (K1) edge node[right, font=\footnotesize] {$k$} (K2);
      \path[->] (R1) edge node[right, font=\footnotesize] {$n$} (R2);
      \draw (0.7,0.9) to (0.7,0.7) to (0.9,0.7);
      \draw (1.7,0.1) to (1.7,0.3) to (1.9,0.3);
    \end{tikzpicture}
  \end{center}
  such that
  $\tuple{f,k}$ is a final pullback complement of
  $\tuple{m,l}$ and the right square is a pushout.
\end{definition}

\begin{proposition}\label{pr-SqPO}
  For every direct SqPO-transformations $\directobj$, $\directobj'$
  with corresponding SqPO-rules $\arule$, $\arule'$, every
  $\rulem:\arule\rightarrow\arule'$ in $\rulecatSqPO$ such that
  $m = m'\circ \rulem_1$, there exists a
  unique $\cat$-morphism $d$ such that
\begin{center}
  \begin{tikzpicture}[xscale=1.9,yscale=1.6,z=-0.7cm]
    \node(G) at (0,0,0){$G$}; \node(G') at (0,0,1){$G'$};
    \node(L) at (0,1,0){$L$};
    \node(L') at (0,1,1){$L'$};
    \node(K) at (1,1,0){$K$};
    \node(K') at (1,1,1){$K'$};
    \node(D) at (1,0,0){$D$};
    \node(D') at (1,0,1){$D'$};
    \path[-](G) edge node[above, font=\footnotesize] {=}(G');
    \path[->](L) edge node[above, font=\footnotesize] {$\rulem_1$}(L');
    \path[->](K) edge node[right, font=\footnotesize] {$\rulem_2$}(K');
    \path[<-,dashed](D) edge node[right, font=\footnotesize] {$d$}(D');
    \path[<-](L) edge node[above, font=\footnotesize] {$l$}(K);
    \path[->](L) edge node[right, near start, font=\footnotesize] {$m$}(G);
    \path[->](L') edge node[left, font=\footnotesize] {$m'$}(G');
    \path[<-](G) edge node[above,near end, font=\footnotesize] {$f$}(D);
    \path[<-](G') edge node[below, font=\footnotesize] {$f'$}(D');
    \path[->](K) edge node[right, font=\footnotesize] {$k$}(D);
    \path[-](L') edge[draw=white, line width=3pt] (K');
    \path[<-](L') edge node[above, font=\footnotesize] {$l'$}(K');
    \path[-](K') edge[draw=white, line width=3pt] (D');
    \path[->](K') edge node[left,near end, font=\footnotesize] {$k'$}(D');
  \end{tikzpicture}
\end{center}
commutes.  
\end{proposition}

Here the existence of $d$ means not only that $\arule'$ removes at
least as much as its subrule $\arule$, but also that it makes at least
as many copies of the items of $G$. Note that when, among two
simultaneous transformations, one makes $p$ copies of an item and the
other makes $q$ copies of the same item, the global context must
contain $pq$ copies of this item, \emph{unless} there is a subsumption
morphism between them. In such a case all the copies made by the
subsumed transformation are simply merged with those made by the
subsuming one (as witnessed by Proposition~\ref{pr-subsume}). Hence
the necessary symmetry between the first and last steps of the Global
Coherent Transformation.

It is then easy to define the category $\directcatSqPO$ of direct
SqPO-transformations, the category $\directcatSqPOm$ of direct
SqPO-transformations with monic matches and faithful functors
$\SqPORF:\directcatSqPO\rightarrow \rulecatSqPO$ and
$\SqPOmRF:\directcatSqPOm\rightarrow \rulecatSqPOm$, as in
Definition~\ref{def-transfoDPO}. We leave this to the reader.

\begin{proposition}\label{pr-SqPO-lf}
  In the category of graphs $\SqPOmRF$ is right-full.
\end{proposition}

Another notion of subrule in the Sesqui-Pushout approach can be found
in \cite[Definition 8]{Lowe15}, where a rule $\arule'$ is defined as a
$\tuple{\rulem_1,\rulem_3}$-extension of $\arule$ if two conditions
are met. The first is that $\rulem_3\circ\arule = \arule'\circ \rulem_1$,
where $\rulem_1$ stands for the span
$L\xleftarrow{\id{L}}L\xrightarrow{\rulem_1} L'$ (and similarly for
$\rulem_3$) and $\circ$ is the
standard composition of spans (using pullbacks, see \cite[Definition
3]{Lowe15}). The products $\rulem_3\circ\arule$, $\arule'\circ
\rulem_1$ yield
\begin{center}
  \begin{tikzpicture}[scale=1.3]
    \node(L1) at (-2,0){$L$}; \node(M1) at (0,0){$K$}; \node(R1) at (2,0){$R'$}; 
    \node(U1) at (-1,1){$K$}; \node(V1) at (1,1){$R$};\node(T1) at (0,2){$R$};
    \path[->] (U1) edge node[above left, font=\footnotesize] {$l$} (L1); 
    \path[<-] (U1) edge node[below left, font=\footnotesize] {$\id{K}$} (M1); 
    \path[->] (V1) edge node[above right, font=\footnotesize] {$\rulem_3$} (R1); 
    \path[<-] (V1) edge node[below right, font=\footnotesize] {$r$} (M1); 
    \path[<-] (T1) edge node[above left, font=\footnotesize] {$r$} (U1); 
    \path[<-] (T1) edge node[above right, font=\footnotesize] {$\id{R}$} (V1); 
    \path[->] (M1) edge (L1); 
    \path[->] (M1) edge (R1);
    \draw (-0.2,0.3) to (0,0.5) to (0.2,0.3);
    \node(L2) at (-2+6,0){$L$}; \node(M2) at (0+6,0){$K$}; \node(R2) at (2+6,0){$R'$}; 
    \node(U2) at (-1+6,1){$L$}; \node(V2) at (1+6,1){$K'$};\node(T2) at (0+6,2){$L'$}; 
    \path[->] (U2) edge node[above left, font=\footnotesize] {$\id{L}$} (L2); 
    \path[<-] (U2) edge node[below left, font=\footnotesize] {$$} (M2); 
    \path[->] (V2) edge node[above right, font=\footnotesize] {$r'$} (R2); 
    \path[<-] (V2) edge node[below right, font=\footnotesize] {$\rulem_2$} (M2); 
    \path[<-] (T2) edge node[above left, font=\footnotesize] {$\rulem_1$} (U2); 
    \path[<-] (T2) edge node[above right, font=\footnotesize] {$l'$} (V2); 
    \path[->] (M2) edge (L2); 
    \path[->] (M2) edge (R2);
    \draw (-0.2+6,0.3) to (0+6,0.5) to (0.2+6,0.3);
  \end{tikzpicture}
\end{center}
hence the equality between these two bottom spans is equivalent to the
existence of $\tuple{\rulem_1,\rulem_2,\rulem_3}:\arule\rightarrow
\arule'$, i.e. that the left square in Definition \ref{def-SqPO} is a
pullback and the right square commutes. This means that any extension
of a rule according to \cite[Definition 8]{Lowe15} subsumes this rule
according to Definition~\ref{def-SqPO}. The converse is false since
the extension requires a second condition, namely that
$\tuple{\rulem_1,l}$ has a final pullback complement. This ensures
that the extension can be decomposed as a product of two spans
\cite[Proposition 9]{Lowe15}, but this is relevant to sequential
rewriting and not to the present notion of subsumption.

\subsection{Pullback-Pushouts}
We next consider the case of PBPO-rules \cite{CorradiniDEPR19}, that
also enables copies of parts of $G$ but with better control of the way
they are linked together and to the rest of $G$. The drawback is that
matchings of the left-hand side of a rule into $G$ should be completed
with a co-match from $G$ to a given ``type'' of the left-hand side.

\begin{definition}[category $\pbporulecat$, direct PBPO-transformations]
  A \emph{PBPO-rule} $\arule$ in $\cat$ is a commuting diagram
  \begin{center}
    \begin{tikzpicture}[xscale=1.7,yscale=1.5]
      \node (L) at (1,1) {$L$};
      \node (K) at (2,1) {$K$};  \node (R) at (3,1) {$\Rc$};  
      \node (TL) at (1,0) {$T_L$};
      \node (TK) at (2,0) {$T_K$};  \node (TR) at (3,0) {$T_R$};  
      \path[->] (K) edge node[above, font=\footnotesize] {$l$} (L);
      \path[->] (K) edge node[above, font=\footnotesize] {$r$} (R);
      \path[->] (TK) edge node[below, font=\footnotesize] {$u$} (TL);
      \path[->] (TK) edge node[below, font=\footnotesize] {$v$} (TR);
      \path[->] (L) edge node[left, font=\footnotesize] {$t_L$} (TL);
      \path[->] (K) edge node[right, font=\footnotesize] {$t_K$} (TK);
      \path[->] (R) edge node[right, font=\footnotesize] {$t_R$} (TR);
    \end{tikzpicture}
  \end{center}
  A morphism
  $\rulem:\arule\rightarrow \arule'$ is a 5-tuple
  $\tuple{\rulem_1, \rulem_2, \rulem_3, \rulem_4, \rulem_5}$ of
  $\cat$-morphisms such that
  \begin{center}
    \begin{tikzpicture}[xscale=1.9,yscale=1.7,z=-0.7cm]
      \node (L) at (1,1,0) {$L$};
      \node (K) at (2,1,0) {$K$};  \node (R) at (3,1,0) {$R$};  
      \node (TL) at (1,0,0) {$T_L$};
      \node (TK) at (2,0,0) {$T_K$};  
      \path[->] (K) edge node[above, font=\footnotesize] {$l$} (L);
      \path[->] (K) edge node[above, font=\footnotesize] {$r$} (R);
      \path[->] (TK) edge node[below, font=\footnotesize] {$u$} (TL);
      \path[->] (L) edge node[near start,right, font=\footnotesize] {$t_L$} (TL);
      \path[->] (K) edge node[near start,right, font=\footnotesize] {$t_K$} (TK);
      \node (L1) at (1,1,1) {$L'$};
      \node (K1) at (2,1,1) {$K'$};  \node (R1) at (3,1,1) {$R'$};  
      \node (TL1) at (1,0,1) {$T_{L'}$};
      \node (TK1) at (2,0,1) {$T_{K'}$};
      \path[-](K1) edge[draw=white, line width=3pt] (L1);
      \path[->] (K1) edge node[above, font=\footnotesize] {$l'$} (L1);
      \path[-](K1) edge[draw=white, line width=3pt] (R1);
      \path[->] (K1) edge node[above, font=\footnotesize] {$r'$} (R1);
      \path[->] (TK1) edge node[below, font=\footnotesize] {$u'$} (TL1);
      \path[->] (L1) edge node[left, font=\footnotesize] {$t_{L'}$} (TL1);
      \path[-](K1) edge[draw=white, line width=3pt] (TK1);
      \path[->] (K1) edge node[near end,left, font=\footnotesize] {$t_{K'}$} (TK1);
      \path[->] (L) edge node[left, font=\footnotesize] {$\rulem_1$} (L1);
      \path[->] (K) edge node[left, font=\footnotesize] {$\rulem_2$} (K1);
      \path[->] (R) edge node[right, font=\footnotesize] {$\rulem_3$} (R1);
      \path[->] (TL1) edge node[left, font=\footnotesize] {$\rulem_4$} (TL);
      \path[->] (TK1) edge node[right, font=\footnotesize] {$\rulem_5$} (TK);
    \end{tikzpicture}
  \end{center}
  commutes. Let $\pbporulecat$ be the category of PBPO-rules on $\cat$
  and their morphisms, with obvious composition and identities.

  A \emph{direct PBPO-transformation} in  $\cat$ is a commuting diagram
  \begin{center}
    \begin{tikzpicture} [xscale=1.7,yscale=1.5]
      \node (L) at (1,2) {$L$};
      \node (K) at (2,2) {$K$};  \node (R) at (3,2) {$R$};  
      \node (G) at (1,1) {$G$};
      \node (D) at (2,1) {$D$};  \node (H) at (3,1) {$H$};  
      \node (TL) at (1,0) {$T_L$};
      \node (TK) at (2,0) {$T_K$};  \node (TR) at (3,0) {$T_R$};  
      \path[->] (L) edge node[right, font=\footnotesize] {$m$} (G);
      \path[->] (K) edge node[right, font=\footnotesize] {$k$} (D);
      \path[->] (R) edge node[right, font=\footnotesize] {$n$} (H);
      \path[->] (D) edge node[below, font=\footnotesize] {$f$} (G);
      \path[->] (D) edge node[below, font=\footnotesize] {$g$} (H);
      \path[->] (G) edge node[right, font=\footnotesize] {$t_G$} (TL);
      \path[->] (D) edge node[right, font=\footnotesize] {$t_D$} (TK);
      \path[->] (H) edge node[right, font=\footnotesize] {$t_H$} (TR);
      \path[->] (K) edge node[above, font=\footnotesize] {$l$} (L);
      \path[->] (K) edge node[above, font=\footnotesize] {$r$} (R);
      \path[->] (TK) edge node[below, font=\footnotesize] {$u$} (TL);
      \path[->] (TK) edge node[below, font=\footnotesize] {$v$} (TR);
      \path[->,bend right=35] (L) edge node[near start,left, font=\footnotesize] {$t_L$} (TL);
      \path[-,bend right=35](K) edge[draw=white, line width=3pt] (TK);
      \path[->,bend right=35] (K) edge node[near start,left, font=\footnotesize] {$t_K$} (TK);
      \path[-,bend right=35](R) edge[draw=white, line width=3pt] (TR);
      \path[->,bend right=35] (R) edge node[near start,left, font=\footnotesize] {$t_R$} (TR);      
      \draw (1.7,0.9) to (1.7,0.7) to (1.9,0.7);
      \draw (2.7,1.1) to (2.7,1.3) to (2.9,1.3);
    \end{tikzpicture}
  \end{center}
  with lower left pullback and upper right pushout.
\end{definition}

To every direct PBPO-transformation obviously corresponds a PBPO-rule
and a partial transformation.

\begin{proposition}\label{pr-pbpo}
  For every direct PBPO-transformations $\directobj$, $\directobj'$
  with corresponding PBPO-rules $\arule$, $\arule'$, every
  $\rulem:\arule\rightarrow\arule'$ in $\pbporulecat$ such that
  $m = m'\circ \rulem_1$ and
  $t_G=\rulem_4\circ t_{G'}$, there exists a
  unique $\cat$-morphism $d$ such that
  \begin{center}
    \begin{tikzpicture}[xscale=1.9,yscale=1.7,z=-0.7cm]
      \node (L) at (1,1,0) {$L$};
      \node (K) at (2,1,0) {$K$};  
      \node (G) at (1,0,0) {$G$};
      \node (D) at (2,0,0) {$D$};  
      \node (TL) at (1,-1,0) {$T_L$};
      \node (TK) at (2,-1,0) {$T_K$};  
      \path[->] (K) edge node[above, font=\footnotesize] {$l$} (L);
      \path[->] (TK) edge node[below, font=\footnotesize] {$u$} (TL);
      \path[->] (D) edge node[above,near start, font=\footnotesize] {$f$} (G);
      \path[->] (L) edge node[near start,right, font=\footnotesize] {$m$} (G);
      \path[->] (K) edge node[right, font=\footnotesize] {$k$} (D);
      \path[->] (G) edge node[left, font=\footnotesize] {$t_G$} (TL);
      \path[->] (D) edge node[right, font=\footnotesize] {$t_D$} (TK);
      \node (L1) at (1,1,1) {$L'$};
      \node (K1) at (2,1,1) {$K'$};
      \node (G1) at (1,0,1) {$G'$};
      \node (D1) at (2,0,1) {$D'$};
      \node (TL1) at (1,-1,1) {$T_{L'}$};
      \node (TK1) at (2,-1,1) {$T_{K'}$};
      \path[-](D1) edge[draw=white, line width=3pt] (G1);
      \path[->] (D1) edge node[below,near end, font=\footnotesize] {$f'$} (G1);
      \path[-](K1) edge[draw=white, line width=3pt] (L1);
      \path[->] (K1) edge node[above, font=\footnotesize] {$l'$} (L1);
      \path[->] (TK1) edge node[below, font=\footnotesize] {$u'$} (TL1);
      \path[->] (L1) edge node[left, font=\footnotesize] {$m'$} (G1); 
      \path[->] (G1) edge node[left, font=\footnotesize] {$t_{G'}$} (TL1); 
      \path[-](K1) edge[draw=white, line width=3pt] (D1); 
      \path[->] (K1) edge node[right, font=\footnotesize] {$k'$} (D1); 
      \path[-](D1) edge[draw=white, line width=3pt] (TK1); 
      \path[->] (D1) edge node[near end,left, font=\footnotesize] {$t_{D'}$} (TK1); 
      \path[->] (L) edge node[left, font=\footnotesize] {$\rulem_1$} (L1);
      \path[->] (K) edge node[left, font=\footnotesize] {$\rulem_2$} (K1);
      \path[->] (TL1) edge node[left, font=\footnotesize] {$\rulem_4$} (TL);
      \path[->] (TK1) edge node[right, font=\footnotesize] {$\rulem_5$} (TK);
      \path[-] (G) edge node[left, font=\footnotesize] {$=$} (G1);
      \path[->,dashed] (D1) edge node[right, font=\footnotesize] {$d$} (D);
    \end{tikzpicture}
  \end{center}
commutes.  
\end{proposition}

We leave it to the reader to define a Rewriting Environment for
PBPO-rules and transformations, with a right-full faithful functor
$\PBPORF: \directcatPBPO\rightarrow \rulecatPBPO$ (provided $\cat$ has
pushouts and pullbacks).

\section{Conclusion and Future Work}\label{sec-conclusion}

Global Coherent Transformations are built from partial transformations
in a way pertaining both to Parallel Coherent Transformations
\cite{BdlTE21a}, by the use of limits on local contexts, and to Global
Transformations \cite{MaignanS15} by applying categories of
rules. The partial transformations involved in a Global Coherent
Transformation are extracted from a Rewriting Environment that provide
a category of rules and a corresponding category of direct
transformations. Their morphisms can be understood as subsumptions due
to Property~\ref{pr-subsume}, i.e., that any subsumed transformation
as defined by a morphism removes or adds nothing more than the
subsuming transformation. This is valid even when rules are able to
make multiple copies of parts of the input.

We have provided Rewriting Environments for the most common approaches
to algebraic rewriting, except the Single Pushout \cite{Lowe93}, which
will be done in a future paper (where we will see that the
interface and right-hand side provided in a partial transformation are
not necessarily those of the applied rule).  We also intend to show
that Global Transformations can be obtained as Global Coherent
Transformations in a suitable environment (except when $\stratcat$ is
empty). Expressiveness of Global Coherent Transformations should be
investigated further, and possibly enhanced.

The notion of Rewriting Environment is as simple as required to define
Global Coherent Transformations, but does not guarantee some
properties that the user might reasonably expect. In particular it
does not prevent the categories $\rulecat$ and $\directcat$ from being
discrete. Of course this is correct if no subsumption is possible, but
is there a way to characterize such properties? It may also seem
strange that, through $\partialcat$, rules are not assumed to have
left-hand sides and direct transformations are not assumed to use
matchings. Thus we may need to enhance Rewriting Environments with a
notion of matching in order to better understand their structure.  We
also need to further analyze the properties of the Rewriting
Environments in Section~\ref{sec-env}: when $\cat$ is an adhesive
category it is an open question whether $\DPOmRF$ is right-full.

\noindent\textbf{Acknowledgements} We thank Rachid Echahed for helpful
discussions and an anonymous reviewer in particular for suggesting the
generalization of Proposition~\ref{pr-subsume}.


\section*{Appendix: Proofs}

\begin{proof}[Proof of Proposition \ref{prop-gc}]
  \emph{If $\rulem:\arule\rightarrow\arule'$ in $\rulecatDPO$ such
    that $\rulem_1$ is monic and $m':L'\rightarrow G$ satisfies the
    gluing condition for $\arule'$ then so does $m'\circ \rulem_1:
    L\rightarrow G$ for $\arule$.}
  
  We use the fact that the pullback $K$ of $l',\,\rulem_1$ is
  isomorphic to en equalizer in $L\times K'$.
  \begin{itemize}
  \item[(GC1)] Let $x$ be an item in $L$ such that
    $m'\circ\rulem_1(x)$ is marked for removal for $\arule$, hence
    such that $x$ has no preimage by $l$, and let $x'$ in $L$ such
    that $m'\circ\rulem_1(x)=m'\circ \rulem_1(x')$.
    If $\rulem_1(x)$ had a preimage $y$ by $l'$ then $x$ and $y$ would
    have a common preimage in the pullback $K$, a contradiction. Hence
    $\rulem_1(x)$ has no preimage by $l'$ so that $m'(\rulem_1(x))$ is
    marked for removal by $m'$, hence $\rulem_1(x) = \rulem_1(x')$ by
    the (GC1) for $m',\,\arule'$, hence $x=x'$.
  \item[(GC2)] Let $v$ be a vertex of $L$ that has no preimage by
  $l$ and is adjacent to an edge $e$ in $L$, then as
  above $\rulem_1(v)$ has no preimage by $l'$. If $e$ had a
  preimage $e'$ by $l$ then
  $l'\circ\rulem_2(e') = \rulem_1\circ l(e') =
  \rulem_1(e)$, i.e., $\rulem_1(e)$ would have a preimage by
  $l'$ in contradiction with (GC2) for
  $m',\,\arule'$. Hence $m'\circ\rulem_1(e)$ is marked for removal by
  $m'\circ\rulem_1$ for $\arule'$. 
  \end{itemize}
\end{proof}

\begin{proof}[Proof of Proposition \ref{prop-DPO}]
  \emph{If $\rulem:\arule\rightarrow\arule'$ in $\rulecatDPO$,
  $m':L'\rightarrow G$ and
  $m'\circ \rulem_1:L\rightarrow G$ have pushout complements as
  below, then there is a unique $d$ such that
  \begin{center}
    \begin{tikzpicture}[xscale=1.9,yscale=1.6,z=-0.7cm]
      \node(TK) at (0,0,0){$G$}; \node(TK') at (0,0,1){$G$};
      \node(K) at (0,1,0){$L$};
      \node(K') at (0,1,1){$L'$};
      \node(R) at (1,1,0){$K$};
      \node(R') at (1,1,1){$K'$};
      \node(TR) at (1,0,0){$D$};
      \node(TR') at (1,0,1){$D'$};
      \path[-](TK) edge node[above, font=\footnotesize] {=}(TK');
      \path[->](K) edge node[above, font=\footnotesize] {$\rulem_1$}(K');
      \path[->](R) edge node[right, font=\footnotesize] {$\rulem_2$}(R');
      \path[<-<,dashed](TR) edge node[right, font=\footnotesize] {$d$}(TR');
      \path[<-<](K) edge node[above, font=\footnotesize] {$l$}(R);
      \path[->](K) edge node[right, near start, font=\footnotesize] {$$}(TK);
      \path[->](K') edge node[left, font=\footnotesize] {$m'$}(TK');
      \path[<-<](TK) edge node[above,near end, font=\footnotesize] {$f$}(TR);
      \path[<-<](TK') edge node[below, font=\footnotesize] {$f'$}(TR');
      \path[->](R) edge node[right, font=\footnotesize] {$k$}(TR);
      \path[-](K') edge[draw=white, line width=3pt] (R');
      \path[<-<](K') edge node[above, font=\footnotesize] {$l'$}(R');
      \path[-](R') edge[draw=white, line width=3pt] (TR');
      \path[->](R') edge node[left,near end, font=\footnotesize] {$k'$}(TR');
    \end{tikzpicture}
  \end{center}
  commutes.}

  The front and back faces are pushouts. For all item $x$ in
  $D'$, $f'(x)$ is not marked for removal by
  $m'$ and we show that is also the case by
  $m'\circ\rulem_1$. Suppose otherwise, then $f'(x)$ has a preimage
  $y$ by $m'\circ\rulem_1$ that has no preimage by
  $l$. However, $\rulem_1(y)$ has a preimage
  $y'$ by
  $l'$, and since the top face is a pullback there should be a
  common preimage of $y$ and $y'$ in
  $K$, a contradiction. Thus we let
  $d(x)$ be the unique preimage of $f'(x)$ by $f$, so that $d$ is
  unique such that $f\circ
  d=f'$. We easily see that $f\circ k = f\circ d\circ k'\circ
  \rulem_2 $ hence the right face of the cube commutes.
\end{proof}

\begin{proof}[Proof of Proposition \ref{pr-h}]
  \emph{For every $\directmorph:\directobj\rightarrow \directobj'$ in
  $\stratcat$ there exists a unique $h_{\directmorph}$ such that
  \begin{center}
    \begin{tikzpicture}[xscale=2.5,yscale=0.8]
      \node(C) at (0,0){$\interf$};
      \node(H) at (1,1){$H_{\directobj}$};
      \node(H') at (1,-1){$H_{\directobj'}$};
      \node(R) at (2,1){$\Rc_{\directobj}$};
      \node(R') at (2,-1){$\Rc_{\directobj'}$};
      \path[->] (C) edge node[above, font=\footnotesize] {$h_{\directobj}$} (H);
      \path[->] (C) edge node[below, font=\footnotesize] {$h_{\directobj'}$} (H');
      \path[->] (R) edge node[above, font=\footnotesize] {$n_{\directobj}$} (H);
      \path[->] (R') edge node[below, font=\footnotesize] {$n_{\directobj'}$} (H');
      \path[->] (R) edge node[right, font=\footnotesize] {$\projr\directmorph$} (R');
      \path[->,dashed] (H) edge node[right, font=\footnotesize] {$h_{\directmorph}$} (H');
    \end{tikzpicture}
  \end{center}
  commutes.}
  
  Since $ \intercone \circ c_{\directobj} = \cone{\directobj} =
  \cone{\directobj'}\circ \projk\directmorph = \intercone \circ
  c_{\directobj'}\circ \projk\directmorph$ then by the unicity of
  $c_{\directobj}$ the left face of the following cube commutes.
  \begin{center}
    \begin{tikzpicture}[xscale=1.9,yscale=1.6,z=-0.7cm]
      \node(C) at (0,0,0){$\interf$}; \node(C') at (0,0,1){$\interf$};
      \node(K) at (0,1,0){$\Kc_{\directobj}$};
      \node(K') at (0,1,1){$\Kc_{\directobj'}$};
      \node(R) at (1,1,0){$\Rc_{\directobj}$};
      \node(R') at (1,1,1){$\Rc_{\directobj'}$};
      \node(H) at (1,0,0){$H_{\directobj}$};
      \node(H') at (1,0,1){$H_{\directobj'}$};
      \path[-](C) edge node[above, font=\footnotesize] {=}(C');
      \path[->](K) edge node[above, font=\footnotesize] {$\projk\directmorph$}(K');
      \path[->](R) edge node[right, font=\footnotesize] {$\projr\directmorph$}(R');
      \path[->,dashed](H) edge node[right, font=\footnotesize] {$h_{\directmorph}$}(H');
      \path[->](K) edge node[above, font=\footnotesize] {$\rc_{\directobj}$}(R);
      \path[->](K) edge node[right, near start, font=\footnotesize] {$c_{\directobj}$}(C);
      \path[->](K') edge node[left, font=\footnotesize] {$c_{\directobj'}$}(C');
      \path[->](C) edge node[above,near end, font=\footnotesize] {$h_{\directobj}$}(H);
      \path[->](C') edge node[below, font=\footnotesize] {$h_{\directobj'}$}(H');
      \path[->](R) edge node[right, font=\footnotesize] {$n_{\directobj}$}(H);
      \path[-](K') edge[draw=white, line width=3pt] (R');
      \path[->](K') edge node[above, font=\footnotesize] {$\rc_{\directobj'}$}(R');
      \path[-](R') edge[draw=white, line width=3pt] (H');
      \path[->](R') edge node[left,near end, font=\footnotesize] {$n_{\directobj'}$}(H');
    \end{tikzpicture}
  \end{center}
  Since the top and front faces also commute then
  $n_{\directobj'}\circ \projr\directmorph \circ \rc_{\directobj} = h_{\directobj'}
  \circ c_{\directobj}$, and since the back face is a pushout we get
  the result. 
\end{proof}

\begin{proof}[Proof of Proposition \ref{pr-subsume}]
  \emph{If $\stratcat'$ is restricted to $\directobj'$ and
    $\directobj'$ is terminal in $\stratcat$ then $\stratcat$ and
    $\stratcat'$ are globally coherent and
    $H_{\stratcat} \simeq H_{\stratcat'}$.}

  For any $\directobj\in\stratcat$ let $\tm{\directobj}$ be the unique morphism
  $\tm{\directobj}:\directobj\rightarrow \directobj'$. Since   $\tuple{\projd{\tm{\directobj}},\, \projk{\tm{\directobj}},\,
    \projr{\tm{\directobj}}}: \partialF\insertGF\insertSF\directobj
  \rightarrow \partialF\insertGF\insertSF\directobj'$ is a morphism in
  $\partialcat$, then $\fc_{\directobj}\circ \projd{\tm{\directobj}} =
  \fc_{\directobj'}$ and hence $\projd{\tm{\directobj}}:
  \fc_{\directobj'}\rightarrow \fc_{\directobj}$ is a morphism in $\slice{G}$.

  Since $\directobj'$ is initial in $\dual{\stratcat}$ there is a
  unique cone $\intercone$ from
  $\ProjG\directobj' = \fc_{\directobj'}$ to $\ProjG$ (defined by
  $\intercone\directobj = \projd{\tm{\directobj}}$ for all
  $\directobj\in\stratcat$) and any cone $\gamma$ from any
  $f\in\slice{G}$ to $\ProjG$ can be written
  $\gamma = \intercone \circ \gamma\directobj'$, hence $\intercone$ is
  a limit cone of $\ProjG$ (see \cite[Exercise III.4.3]{MacLane}), so that
  $f_{\stratcat} \simeq \fc_{\directobj'}$ and
  $\interf \simeq \Dc_{\directobj'}$.

  Let
  $\cone{\directobj}=\intercone\circ \kc_{\directobj'} \circ
  \projk{\tm{\directobj}}$ (where
  $\projk{\tm{\directobj}} : \fc_{\directobj}\circ \kc_{\directobj}
  \rightarrow \fc_{\directobj'}\circ \kc_{\directobj'}$ and
  $\kc_{\directobj'}: \fc_{\directobj'}\circ \kc_{\directobj'}
  \rightarrow \fc_{\directobj'}$ are morphisms in $\slice{G}$ as
  above), this is a cone from
  $\fc_{\directobj}\circ \kc_{\directobj} $ to $\ProjG$ such that
  $\cone{\directobj}\directobj = \projd{\tm{\directobj}} \circ
  \kc_{\directobj'} \circ \projk{\tm{\directobj}} =
  \kc_{\directobj}$. Besides, for every
  $\directmorph:\directobj_1\rightarrow \directobj_2$ we have
  $\cone{\directobj_1} = \cone{\directobj_2}\circ
  \projk{\directmorph}$ since
  $\tm{\directobj_2}\circ \directmorph = \tm{\directobj_1}$. Hence
  $(\cone{\directobj})_{\directobj\in\stratcat}$ is a coherent system
  of cones for $\stratcat$, which is therefore globally coherent.

  Since $\directobj'$ is terminal in $\stratcat$ there is as above a
  colimit cone from $\ProjH$ to
  $\ProjH\directobj' = h_{\directobj'}: \interf\rightarrow
  H_{\directobj'}$, hence $H_{\stratcat}\simeq H_{\directobj'}$ (the
  pushout of $\rc_{\directobj'}$ and
  $c_{\directobj'}=\kc_{\directobj'}\circ \projk{\tm{\directobj'}} =
  \kc_{\directobj'}$). We finally note that $\directobj'$ is 
  terminal in $\stratcat'$.
\end{proof}

\begin{proof}[Proof of Lemma \ref{lm-IGF-full}]
\emph{$\insertGF$ is a full and right-full embedding.}

  The functor $G:\termcat\rightarrow \discr{\cat}$ is a full embedding hence so is
  $\insertGF$. For all $\directobj'\in \directcatG$,
  $\directobj\in\directcat$ and
  $\directmorph: \directobj \rightarrow \insertGF\directobj'$ we have
  $\inputF\partialF \directobj = \insertGF\inputF\partialF
  \directobj'=G$ hence $\inputF\partialF\directmorph = \id{G}$. Since
  $G$ and $\id{G}$ also have preimages by functor $G$ there must be
  preimages $\directobj'_1\in\directcatG$ and
  $\directmorph_1:\directobj'_1\rightarrow \directobj'$ in
  $\directcatG$ such that $\insertGF\directmorph_1 =\directmorph$,
  hence $\insertGF$ is right-full.
\end{proof}

\begin{proof}[Proof of Proposition \ref{pr-right-full}]
\emph{If $\ruleF$ is right-full (resp. faithful) then so is $\ruleFSG$.}

  For all $\directobj'\in\directcatSG$, $\arule\in\rulesys$ and
  $\rulem:\arule\rightarrow \arule'$, where
  $\arule'=\ruleFSG\directobj'$, we have
  $\insertF\arule' = \ruleF\insertGF\insertSF\directobj'$ and
  $\insertF\rulem:\insertF\arule\rightarrow \insertF\arule'$ in
  $\rulecat$, and since by Lemma~\ref{lm-IGF-full}
  $\ruleF\circ \insertGF$ is right-full then there exists
  $\directobj'_1\in\directcatG$ and
  $\directmorph_1:\directobj'_1\rightarrow \insertSF\directobj'$ such
  that $\ruleF\insertGF\directmorph_1 = \insertF\rulem$. Thus
  $\insertF\arule$ and $\insertF\rulem$ have preimages by $\insertF$
  and $\ruleF\circ \insertGF$, hence they must have preimages
  $\directobj\in\directcatSG$ and
  $\directmorph:\directobj\rightarrow \directobj'$ such that
  $\insertGF\directmorph = \directmorph_1$ and
  $\ruleFSG\directmorph=\rulem$.

  If $\ruleF$ is faithful, since $\insertGF$ is faithful then so is
  $\ruleF\circ \insertGF$, and hence so is $\ruleFSG$. 
\end{proof}

\begin{proof}[Proof of Proposition \ref{pr-adhesive}]
  \emph{If $\cat$ is adhesive, $\directobj,\directobj'\in\directcatDPO$ and
  $\rulem: \DPORF\directobj \rightarrow \DPORF\directobj'$ such that
  $m=m' \circ \rulem_1$ then there exists a unique
  $\directmorph:\directobj\rightarrow \directobj'$ such that
  $\DPORF\directmorph = \rulem$.}
  
  Let $G=\inputF\partialF \directobj= \inputF\partialF \directobj'$, we
  consider the following diagram
  
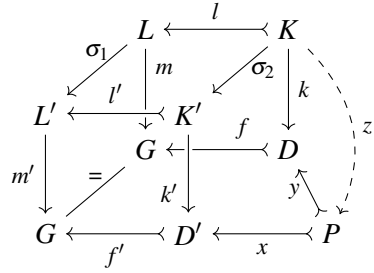
\begin{center}
  \begin{tikzpicture}[xscale=1.9,yscale=1.6,z=-0.7cm]
    \node(G) at (0,0,0){$G$}; \node(G') at (0,0,1){$G$};
    \node(L) at (0,1,0){$L$};
    \node(L') at (0,1,1){$L'$};
    \node(K) at (1,1,0){$K$};
    \node(K') at (1,1,1){$K'$};
    \node(D) at (1,0,0){$D$};
    \node(D') at (1,0,1){$D'$};
    \node(P) at (2,0,1){$P$};
    \path[-](G) edge node[above, font=\footnotesize] {=}(G');
    \path[->](L) edge node[above, font=\footnotesize] {$\rulem_1$}(L');
    \path[->](K) edge node[right, font=\footnotesize] {$\rulem_2$}(K');
    \path[<-<](L) edge node[above, font=\footnotesize] {$l$}(K);
    \path[->](L) edge node[right, near start, font=\footnotesize] {$m$}(G);
    \path[->](L') edge node[left, font=\footnotesize] {$m'$}(G');
    \path[<-<](G) edge node[above,near end, font=\footnotesize] {$f$}(D);
    \path[<-<](G') edge node[below, font=\footnotesize] {$f'$}(D');
    \path[->](K) edge node[right, font=\footnotesize] {$k$}(D);
    \path[-](L') edge[draw=white, line width=3pt] (K');
    \path[<-<](L') edge node[above, font=\footnotesize] {$l'$}(K');
    \path[-](K') edge[draw=white, line width=3pt] (D');
    \path[->](K') edge node[left,near end, font=\footnotesize] {$k'$}(D');
    \path[>->](P) edge node[below, font=\footnotesize] {$x$}(D');
    \path[>->](P) edge node[left, font=\footnotesize] {$y$}(D);
    \path[->,dashed, bend left](K) edge node[right, font=\footnotesize] {$z$}(P);
  \end{tikzpicture}
\end{center}
where the bottom face is a pullback. By \cite[Lemma 4.2]{LackS05}
monics are stable under pushouts hence $f$ and
$f'$ are monics and therefore also $x$ and $y$. By the
commuting properties we have
$f\circ k= f'\circ k'\circ\rulem_2$, hence there exists a unique $z$ such
that $y\circ z = k$ and
$x\circ z = k'\circ \rulem_2$.

The front face is a pushout along the monic $l$, hence it
is a pullback \cite[Lemma 4.3]{LackS05}, as is the top face, hence by
composition the square formed by $l$, $m$,
$f'$, $k'\circ\rulem_2$ is also a pullback.

The back face is a pushout along the monic $l$, hence it
is a VK-square and bottom face of the
commuting cube below
  
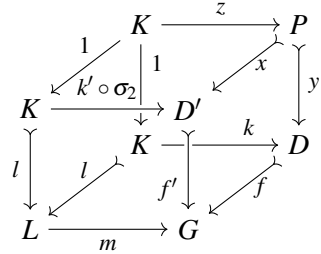
\begin{center}
    \begin{tikzpicture}[xscale=2.1,yscale=1.6,z=-0.7cm]
      \node(Kb) at (0,0,0){$K$};
      \node(L) at (0,0,1){$L$};
      \node(K) at (0,1,0){$K$};
      \node(K') at (0,1,1){$K$};
      \node(P) at (1,1,0){$P$};
      \node(D') at (1,1,1){$D'$};
      \node(D) at (1,0,0){$D$};
      \node(G) at (1,0,1){$G$};
      \path[>->](Kb) edge node[above, font=\footnotesize] {$l$}(L);
      \path[->](K) edge node[above, font=\footnotesize] {$\id{}$}(K');
      \path[>->](P) edge node[right, font=\footnotesize] {$x$}(D');
      \path[>->](D) edge node[right, font=\footnotesize] {$f$}(G);
      \path[->](K) edge node[above, font=\footnotesize] {$z$}(P);
      \path[->](K) edge node[right, near start, font=\footnotesize] {$\id{}$}(Kb);
      \path[>->](K') edge node[left, font=\footnotesize] {$l$}(L);
      \path[->](Kb) edge node[above,near end, font=\footnotesize] {$k$}(D);
      \path[->](L) edge node[below, font=\footnotesize] {$m$}(G);
      \path[>->](P) edge node[right, font=\footnotesize] {$y$}(D);
      \path[-](K') edge[draw=white, line width=3pt] (D');
      \path[->](K') edge node[above, font=\footnotesize] {$k'\circ \rulem_2$}(D');
      \path[-](D') edge[draw=white, line width=3pt] (G);
      \path[>->](D') edge node[left,near end, font=\footnotesize] {$f'$}(G);
    \end{tikzpicture}
  \end{center}
  Its front and right faces are pullbacks. Since $l$ is
  monic then its left face is a pullback, and since $y$ is monic its
  back face is also a pullback. Hence its top face is a pushout, and
  since isomorphisms are preserved by pushouts, $x$ is an
  isomorphism.

  Let $d=y \circ \invf{x}$, we see that $f \circ d = f'$ and
  $d\circ k' \circ \rulem_2 = y \circ z = k$, so that
  $\directmorph=\tuple{\rulem_1,\rulem_2,\rulem_3,d}$ is a morphism
  from $\directobj$ to $\directobj'$ in $\directcatDPO$ such that
  $\DPORF\directmorph = \rulem$. Its unicity is obvious.
\end{proof}

\begin{proof}[Proof of Proposition \ref{pr-SqPO}]
  \emph{For every direct SqPO-transformations $\directobj$, $\directobj'$
  with corresponding SqPO-rules $\arule$, $\arule'$, every
  $\rulem:\arule\rightarrow\arule'$ in $\rulecatSqPO$ such that
  $m = m'\circ \rulem_1$, there exists a
  unique $d$ such that
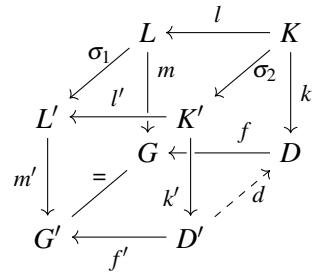
\begin{center}
  \begin{tikzpicture}[xscale=1.9,yscale=1.6,z=-0.7cm]
    \node(G) at (0,0,0){$G$}; \node(G') at (0,0,1){$G'$};
    \node(L) at (0,1,0){$L$};
    \node(L') at (0,1,1){$L'$};
    \node(K) at (1,1,0){$K$};
    \node(K') at (1,1,1){$K'$};
    \node(D) at (1,0,0){$D$};
    \node(D') at (1,0,1){$D'$};
    \path[-](G) edge node[above, font=\footnotesize] {=}(G');
    \path[->](L) edge node[above, font=\footnotesize] {$\rulem_1$}(L');
    \path[->](K) edge node[right, font=\footnotesize] {$\rulem_2$}(K');
    \path[<-,dashed](D) edge node[right, font=\footnotesize] {$d$}(D');
    \path[<-](L) edge node[above, font=\footnotesize] {$l$}(K);
    \path[->](L) edge node[right, near start, font=\footnotesize] {$m$}(G);
    \path[->](L') edge node[left, font=\footnotesize] {$m'$}(G');
    \path[<-](G) edge node[above,near end, font=\footnotesize] {$f$}(D);
    \path[<-](G') edge node[below, font=\footnotesize] {$f'$}(D');
    \path[->](K) edge node[right, font=\footnotesize] {$k$}(D);
    \path[-](L') edge[draw=white, line width=3pt] (K');
    \path[<-](L') edge node[above, font=\footnotesize] {$l'$}(K');
    \path[-](K') edge[draw=white, line width=3pt] (D');
    \path[->](K') edge node[left,near end, font=\footnotesize] {$k'$}(D');
  \end{tikzpicture}
\end{center}
commutes.}

  By composition of pullbacks
  $\tuple{k'\circ \rulem_2,l}$ is a pullback of
  $\tuple{f', m}$, and since
  $\tuple{f,k}$ is a final pullback complement of
  $\tuple{m,l}$ then there is a unique
  $d:D' \rightarrow D$ such that
  $f'=f\circ d$ and
  $k = d\circ k'\circ \rulem_2$.
\end{proof}

\begin{proof}[Proof of Proposition \ref{pr-SqPO-lf}]
\emph{In the category of graphs $\SqPOmRF$ is right-full.}

For all $\directobj'\in \directcatSqPOm$ and
$\rulem:\arule\rightarrow \SqPOmRF \directobj'$ in $\rulecatSqPOm$,
the matching $m'\circ\rulem_1:L\rightarrow G$ is monic hence by
\cite[Construction 6]{CorradiniHHK06} $\tuple{m'\circ\rulem_1,l}$ has
a final pullback complement, hence there is a
$\directobj\in \directcatSqPOm$ with $m = m'\circ\rulem_1$ and
$\SqPOmRF\directobj = \arule$, and by Proposition~\ref{pr-SqPO} there
is a (unique) $\directmorph:\directobj \rightarrow \directobj'$ in
$\directcatSqPOm$ such that $\SqPOmRF\directmorph = \rulem$.
\end{proof}

\begin{proof}[Proof of Proposition \ref{pr-pbpo}]
\emph{For every direct PBPO-transformations $\directobj$, $\directobj'$
  with corresponding PBPO-rules $\arule$, $\arule'$, every
  $\rulem:\arule\rightarrow\arule'$ in $\pbporulecat$ such that
  $m = m'\circ \rulem_1$ and
  $t_G=\rulem_4\circ t_{G'}$, there exists a
  unique $d$ such that
  \begin{center}
    \begin{tikzpicture}[xscale=1.9,yscale=1.7,z=-0.7cm]
      \node (L) at (1,1,0) {$L$};
      \node (K) at (2,1,0) {$K$};  
      \node (G) at (1,0,0) {$G$};
      \node (D) at (2,0,0) {$D$};  
      \node (TL) at (1,-1,0) {$T_L$};
      \node (TK) at (2,-1,0) {$T_K$};  
      \path[->] (K) edge node[above, font=\footnotesize] {$l$} (L);
      \path[->] (TK) edge node[below, font=\footnotesize] {$u$} (TL);
      \path[->] (D) edge node[above,near start, font=\footnotesize] {$f$} (G);
      \path[->] (L) edge node[near start,right, font=\footnotesize] {$m$} (G);
      \path[->] (K) edge node[right, font=\footnotesize] {$k$} (D);
      \path[->] (G) edge node[left, font=\footnotesize] {$t_G$} (TL);
      \path[->] (D) edge node[right, font=\footnotesize] {$t_D$} (TK);
      \node (L1) at (1,1,1) {$L'$};
      \node (K1) at (2,1,1) {$K'$};
      \node (G1) at (1,0,1) {$G'$};
      \node (D1) at (2,0,1) {$D'$};
      \node (TL1) at (1,-1,1) {$T_{L'}$};
      \node (TK1) at (2,-1,1) {$T_{K'}$};
      \path[-](D1) edge[draw=white, line width=3pt] (G1);
      \path[->] (D1) edge node[below,near end, font=\footnotesize] {$f'$} (G1);
      \path[-](K1) edge[draw=white, line width=3pt] (L1);
      \path[->] (K1) edge node[above, font=\footnotesize] {$l'$} (L1);
      \path[->] (TK1) edge node[below, font=\footnotesize] {$u'$} (TL1);
      \path[->] (L1) edge node[left, font=\footnotesize] {$m'$} (G1); 
      \path[->] (G1) edge node[left, font=\footnotesize] {$t_{G'}$} (TL1); 
      \path[-](K1) edge[draw=white, line width=3pt] (D1); 
      \path[->] (K1) edge node[right, font=\footnotesize] {$k'$} (D1); 
      \path[-](D1) edge[draw=white, line width=3pt] (TK1); 
      \path[->] (D1) edge node[near end,left, font=\footnotesize] {$t_{D'}$} (TK1); 
      \path[->] (L) edge node[left, font=\footnotesize] {$\rulem_1$} (L1);
      \path[->] (K) edge node[left, font=\footnotesize] {$\rulem_2$} (K1);
      \path[->] (TL1) edge node[left, font=\footnotesize] {$\rulem_4$} (TL);
      \path[->] (TK1) edge node[right, font=\footnotesize] {$\rulem_5$} (TK);
      \path[-] (G) edge node[left, font=\footnotesize] {$=$} (G1);
      \path[->,dashed] (D1) edge node[right, font=\footnotesize] {$d$} (D);
    \end{tikzpicture}
  \end{center}
  commutes. }

  By hypothesis the two front, back, left faces commute, as well as
  the top and bottom faces. Thus \[u\circ \rulem_5\circ
    t_{D'} = \rulem_4\circ u'\circ 
    t_{D'} = \rulem_4\circ t_{G'} \circ f'
    = t_G \circ f',\] and since $D$
  is a pullback then there exists a unique $d$ such that the right and
  top face of the bottom cube commute. This also means that
  $(D,f,t_D)$ is a mono-source, and
  since
  \[\left\{
      \begin{array}{l}
        f\circ d\circ k'\circ\rulem_2 = f' \circ k'\circ\rulem_2 = m'
        \circ l'\circ\rulem_2 = m' \circ\rulem_1\circ l = m \circ l = f
        \circ k \\
        t_D \circ d\circ k' \circ\rulem_2 = \rulem_5\circ t_{D'}\circ
        k'\circ\rulem_2 = \rulem_5\circ t_{K'}\circ\rulem_2 = t_K =
        t_D\circ k 
      \end{array}\right.\]
  then $d\circ k'\circ\rulem_2 = k$.
\end{proof}

\end{document}